\pdfoutput=1
\documentclass[sigconf,screen,nonacm]{acmart}

\usepackage[utf8]{inputenc}
\usepackage{tikz}
\usepackage[bibliography=common]{apxproof}
\usepackage{relsize}
\usepackage{stackengine}
\stackMath

\AtBeginDocument{%
  \providecommand\BibTeX{{%
    \normalfont B\kern-0.5em{\scshape i\kern-0.25em b}\kern-0.8em\TeX}}}

\newtheorem{openpb}{Open problem}
\newtheorem{conjecture}{Conjecture}
\newtheoremrep{theorem}{Theorem}[section]
\newtheoremrep{proposition}[theorem]{Proposition}
\newtheoremrep{lemma}[theorem]{Lemma}
\newtheoremrep{example}[theorem]{Example}

\usepackage{environ}
\makeatletter
\newsavebox{\measure@tikzpicture}
\NewEnviron{scaletikzpicturetowidth}[1]{%
  \def\tikz@width{#1}%
  \def\tikzscale{1}\begin{lrbox}{\measure@tikzpicture}%
  \BODY
  \end{lrbox}%
  \pgfmathparse{#1/\wd\measure@tikzpicture}%
  \edef\tikzscale{\pgfmathresult}%
  \BODY
}
\makeatother
\definecolor{mygreen}{rgb}{0.09, 0.45, 0.27}
\renewcommand\leq\leqslant
\renewcommand\geq\geqslant

\newcommand{\cnf}{\mathrm{CNF}}
\newcommand{\dnf}{\mathrm{DNF}}
\newcommand{\pqe}{\mathrm{PQE}}
\newcommand\dep{\mathsf{DEP}}
\newcommand\eul{\mathsf{e}}
\newcommand\sat{\mathsf{SAT}}
\newcommand\VARS{\mathsf{Vars}}
\newcommand\false{\mathsf{False}}
\newcommand\true{\mathsf{True}}

\newcommand\lin{\mathsf{Lin}}
\newcommand\ddptime{\text{d-D}(\mathrm{PTIME})}
\newcommand\ddnnfptime{\text{d-DNNF}(\mathrm{PTIME})}
\newcommand\ddpsize{\text{d-D}(\mathrm{PSIZE})}

\newcommand\obddptime{\text{OBDD}(\mathrm{PTIME})}
\newcommand\obddpsize{\text{OBDD}(\mathrm{PSIZE})}
\newcommand\fbddpsize{\text{FBDD}(\mathrm{PSIZE})}
\newcommand\fbddptime{\text{FBDD}(\mathrm{PTIME})}

\newcommand\defeq{\stackrel{\mathrm{def}}{=}}
\renewcommand{\Pr}{\mathrm{Pr}}

\newcommand{\squig}{{\scriptstyle\sim\mkern-3.9mu}}

\newcommand{\rsquigend}{{\scriptstyle\rule{.1ex}{0ex}\rhd}}
\newcounter{sqindex}
\newcommand\squigs[1]{%
  \setcounter{sqindex}{0}%
  \whiledo {\value{sqindex}< #1}{\addtocounter{sqindex}{1}\squig}%
}
\newcommand\rewr[2]{%
  \mathbin{\stackon[2pt]{\squigs{#2}\rsquigend}{\scriptscriptstyle\text{#1\,}}}%
}

\begin{document}

\title{Solving a Special Case of the Intensional vs Extensional Conjecture in Probabilistic Databases}


\author{Mikaël Monet}
\affiliation{%
  \institution{DCC, University of Chile \& IMFD Chile}
  \city{Santiago}
  \country{Chile}}
\email{mikael.monet@imfd.cl}

%\renewcommand{\shortauthors}{M. Monet}

\begin{abstract}
  We consider the problem of exact probabilistic inference for Union of Conjunctive Queries (UCQs) on tuple-independent databases. 
For this problem, two approaches currently coexist.
In the \emph{extensional method}, query evaluation is performed by exploiting the structure of the query, and relies heavily on the use of the inclusion--exclusion principle.
In the \emph{intensional method}, one first builds a representation of the \emph{lineage} of the query in a tractable formalism of knowledge compilation.
The chosen formalism should then ensure that the probability can be efficiently computed using simple disjointness and independence assumptions, without the need of performing inclusion--exclusion. 
The extensional approach has long been thought to be strictly more powerful than the intensional approach, the reason being that for some queries, the use of inclusion--exclusion seemed unavoidable.
In this paper we introduce a new technique to construct lineage representations as \emph{deterministic decomposable circuits} in polynomial time. 
We prove that this technique applies to a class of UCQs that had been conjectured to separate the complexity of the two approaches.
In essence, we show that relying on the inclusion--exclusion formula can be avoided by using negation.
This result brings back hope to prove that the intensional approach can handle all tractable~UCQs.

\end{abstract}




\settopmatter{printfolios=true} % From https://tex.stackexchange.com/questions/358088/add-page-number-in-acm-2017-sigconf-template, to get page numbers
\maketitle

\section{Introduction}
  \label{sec:introduction}
  Probabilistic databases~\cite{suciu2011probabilistic} have been introduced in answer to the need to capture data uncertainty and reason about it. 
In their simplest and most common form, probabilistic databases
consist of a relational database in which each tuple is annotated with an independent probability value.
This value is supposed to represent how confident we are about having the tuple in the database. 
This is known as the \emph{tuple-independent database} (TID) model~\cite{fuhr1997probabilistic,zimanyi1997query}.
While a traditional database can only \emph{satisfy} or \emph{violate} a 
Boolean query, a probabilistic database has a certain probability of satisfying it.
Given a Boolean query $Q$, the \emph{probabilistic query evaluation problem for $Q$}
($\pqe(Q)$) then asks for the probability that the query holds on an input probabilistic database.
Here, the complexity of $\pqe(Q)$ is measured as a function of the input database, hence considering that the Boolean query $Q$ is fixed.
This is known as \emph{data complexity}~\cite{vardi1982complexity}, and is motivated by the fact that the queries are usually much smaller than the data.

When $Q$ is a union of conjunctive queries, a dichotomy result is provided by the seminal work of Dalvi and Suciu~\cite{dalvi2012dichotomy}: either $Q$ is \emph{safe}
and $\pqe(Q)$ is in polynomial time (PTIME), or $Q$ is not safe and $\pqe(Q)$ is \#P-hard. 
The algorithm to compute the probability of a safe UCQ exploits the first-order structure of the query
to find a so called \emph{safe extensional query plan}, using extended relational operators that can manipulate probabilities.
An unusual feature of this algorithm is the use of the inclusion--exclusion principle
(more precisely, of the \emph{Möbius inversion formula}), which allows for distinct hard sub-queries to cancel each other.
This approach is referred to as \emph{extensional query evaluation}, or \emph{lifted inference}~\cite{gribkoff2014lifted,poole2003first}.

A second approach to $\pqe$ for safe UCQs is \emph{intensional query evaluation} or \emph{grounded inference}, and consists of two steps. 
First, compute a representation of the \emph{lineage}~\cite{green2006models} (also called \emph{provenance}) of the
query $Q$ on the database $D$, which is a Boolean function intuitively representing which tuples of~$D$ suffice to satisfy~$Q$.
Second, perform weighted model counting on the lineage to obtain the probability.
This is, for instance, the approach taken by the recent probabilistic database management system ProvSQL~\cite{senellart2018provsql}.
To ensure that model counting is tractable, we use the structure of the query to represent the lineage in tractable formalisms from the field of knowledge compilation.
This includes read-once Boolean formulas~\cite{gurvich1977repetition}, free or ordered binary decision diagrams
(FBDDs~\cite{akers1978binary}, OBDDs~\cite{bryant1986graph,wegener2004bdds}), deterministic decomposable normal forms (d-DNNFs~\cite{darwiche2001tractability}),
decision decomposable normal forms (dec-DNNFs~\cite{huang2005dpll, huang2007language}), decomposable logic decision diagrams (DLDDs~\cite{beame2017exact}),
deterministic decomposable circuits (\mbox{d-Ds}~\cite{darwiche2001tractability,suciu2011probabilistic,monet2018towards}\footnote{The notation d-D was introduced in~\cite{monet2018towards}.}),
structured versions thereof~\cite{pipatsrisawat2008new}, etc.
These can all be seen as restricted classes of Boolean circuits, that \emph{by definition} allow for
efficient (in fact, linear) probability computation: the $\land$-gates are \emph{decomposable} (their inputs depend on disjoint sets of variables, hence represent independent probabilistic events) and the $\lor$-gates are \emph{deterministic} (their inputs represent disjoint probabilistic events).
There are many advantages of this approach compared to lifted inference.
First, and this is also true for non-probabilistic evaluation, the lineage can help explain the query answer~\cite{buneman2001why}.
Second, having the lineage in a good knowledge compilation formalism can be useful for various other applications:
we could for instance update the tuples' probabilities and compute the new result easily, or compute the most probable state of
the data that satisfies the query~\cite{darwiche2009modeling, sharma2018knowledge}, or enumerate satisfying states
with constant delay~\cite{amarilli2017circuit}, or produce random samples of satisfying states~\cite{sharma2018knowledge}, etc.
This ability to reuse the lineage for multiple tasks is precisely one of the main motivations of the field of knowledge compilation.

A natural question is then to ask if the intensional
approach is as powerful as the extensional one.
To answer it, a whole line of research started, whose goal is to determine exactly which queries can be handled by which formalisms of knowledge compilation.
It is known, for example, that the UCQs whose lineages have polynomial-sized read-once formulas representations are
exactly the \emph{hierarchical-read-once} UCQs~\cite{olteanu2008using, jha2013knowledge}, and that
those having polynomial-sized OBDDs are exactly the \emph{inversion-free} UCQs~\cite{jha2013knowledge}.
(This later result on OBDDs was extended to UCQs with disequalities atoms~\cite{jha2012tractability} and to self-join--free CQs with negations~\cite{fink2016dichotomies}.)
However, as far as we know, the characterization is open for all the other formalisms of knowledge compilation.
What we call the \emph{intensional--extensional conjecture}, formulated
in~\cite{jha2013knowledge, suciu2011probabilistic, dalvi2012dichotomy}, states that, for safe queries, extensional query evaluation
is strictly more powerful than the 
knowledge compilation approach.
Or, in other words, that there exists a UCQ which is safe (i.e., that can be handled by the extensional approach) whose lineages on arbitrary databases
cannot be computed in PTIME in a tractable knowledge compilation formalism (i.e., cannot be handled by the intensional approach). 
As we have already mentioned, the conjecture depends on which tractable formalism we consider.
However, generally speaking, the idea would be that knowledge compilation formalisms cannot simulate efficiently the Möbius inversion
formula used in the algorithm for safe queries.

The conjecture has recently been shown in~\cite{beame2017exact} to hold for the formalism of DLDDs (including OBDDs, FBDDs and dec-DNNFs),
which captures the execution of modern model counting algorithms by restricting the power of determinism. 
Another independent result \cite{bova2017circuit} shows that the conjecture also holds when we consider the class
of \emph{structured} d-DNNFs (d-SDNNFs, also including OBDDs), which are d-DNNFs that follow the structure of a \emph{v-tree}~\cite{pipatsrisawat2008new}.
However the question is still widely open for the most expressive formalisms, namely, d-DNNFs and \mbox{d-Ds}. 
Indeed, it could be the case that the conjecture does not hold for such expressive formalisms, which would imply that
we could explain the tractability of all safe UCQs via knowledge compilation, and thus enjoy all of its advantages.

In this paper we focus on a class of Boolean queries that have been intensively studied
already~\cite{dalvi2012dichotomy,jha2013knowledge,beame2017exact,bova2017circuit,monet2018towards,suciu2011probabilistic},
and that we name here the \emph{$\mathcal{H}$-queries}.
An $\mathcal{H}$-query can be defined as a Boolean combination of very simple CQs.
To ease notation, we can always denote an $\mathcal{H}$-query as $Q_\varphi$, where~$\varphi$ is a Boolean function whose variables correspond
to simple CQs (see Section~\ref{sec:review} for the exact definition).
Hence, when the Boolean function~$\varphi$ is monotone, $Q_\varphi$ is a UCQ, and we write $\mathcal{H}^+$ the set of $\mathcal{H}$-queries that are UCQs.
The safe $\mathcal{H}^+$-queries were conjectured in~\cite{suciu2011probabilistic,jha2013knowledge, dalvi2012dichotomy} to not have tractable lineage
representations in any good knowledge compilation formalism.
In fact, the $\mathcal{H}^+$-queries are precisely the ones that were used to prove the intensional--extensional conjecture for both
DLDDs~\cite{beame2017exact} and d-SDNNFs~\cite{bova2017circuit}, leaving little hope for the remaining formalisms.
In a sense, these queries are the simplest UCQs that illustrate the need of the Möbius inversion formula in Dalvi and Suciu's algorithm.
Hence, they are also the first serious obstacle to disproving the conjecture for the most expressive formalisms of knowledge compilation.

\paragraph{\bf Contributions, short version.}

Building on preliminary investigations~\cite{monet2018towards}, we develop a new technique to construct \mbox{d-Ds} in polynomial time for some of the $\mathcal{H}$-queries, and we prove that our technique applies to \emph{all the safe $\mathcal{H}^+$-queries}.
What is more, we also show that this technique applies to some $\mathcal{H}$-queries that are not UCQs (i.e., in $\mathcal{H} \setminus \mathcal{H}^+$), and we show \#P-hardness for a subset of those for which it does not work.
A picture of the situation can be found in Figure~\ref{fig:full-picture}.
Not illustrated in Figure~\ref{fig:full-picture}, we also provide a detailed analysis of when lineages for an $\mathcal{H}$-query $Q_\varphi$ can be transformed into lineages for another $\mathcal{H}$-query $Q_{\varphi'}$, with only a polynomial increase in size.

\paragraph{\bf Contributions, long version.}
We now explain these results and our methodology more exhaustively. 
The dichotomy theorem of Dalvi and Suciu implies that an $\mathcal{H}^+$-query $Q_\varphi$ is safe if
and only if the Möbius value of the \emph{CNF lattice} of~$\varphi$ is zero
(and for $Q_\varphi \in \mathcal{H} \setminus \mathcal{H}^+$ we do not know, because these are not UCQs).
Our starting point is to reformulate that criterion by showing that this value is equal to the \emph{Euler characteristic} of~$\varphi$. 
This connection seems to have been unnoticed so far in the probabilistic database literature.
Hence, we obtain that $Q_\varphi \in \mathcal{H}^+$ is safe if and only if the Euler characteristic of
$\varphi$ is zero.
While it is not clear how to define the CNF lattice for the $\mathcal{H}$-queries that are not UCQs, one advantage
of our characterization is that the Euler characteristic is defined for any Boolean function~$\varphi$.

We then use this new characterization to show that all $\mathcal{H}$-queries $Q_\varphi$ for which~$\varphi$ has zero Euler characteristic have d-D representation of their lineages constructible in polynomial time.
This can be considered as the main result of this article.
Its proof contains three ingredients:
\begin{itemize}
	\item The first one is a result on constructing OBDDs for restricted fragments of UCQs with negations~\cite{fink2016dichotomies}, that we use as a black box
		(but we explain the relevant parts in Appendix, for completeness).
	\item The second one is a notion of what we call a \emph{fragmentable Boolean function}, intuitively designed to ensure that $Q_\varphi$
		has \mbox{d-Ds} constructible in PTIME whenever~$\varphi$ is fragmentable, and that relies on the result of~\cite{fink2016dichotomies}.
	\item The third component is a certain transformation between Boolean functions, where~$\varphi$ can be transformed into
		$\varphi'$ by iteratively (1) removing from~$\varphi$ two connected (i.e., that differ in only one variable) satisfying
		valuations; or (2) adding to~$\varphi$ two connected valuations that did not satisfy~$\varphi$.
\end{itemize}
This transformation defines an equivalence relation $\simeq$ between all Boolean functions,
and we can show that if $\varphi \simeq \bot$ (the Boolean function that is always false), then~$\varphi$ is fragmentable.
We then show that when~$\varphi$ has zero Euler characteristic it is equivalent to~$\bot$, thus completing the proof.
We also justify the definition of our transformation by showing that using only (1) or only (2) is not enough.

Our next main result is to show that if~$\varphi$ and~$\varphi'$ have the same Euler characteristic, then (a) $\pqe(Q_\varphi)$ and $\pqe(Q_{\varphi'})$ are reducible to each other (under PTIME Turing reductions); (b) we can compute in polynomial time d-D representations of the lineage of~$Q_{\varphi}$ on arbitrary databases iff we can do the same for~$Q_{\varphi'}$;
and (c) the lineages of $Q_\varphi$ can be represented as \mbox{d-Ds} of polynomial size iff those of $Q_{\varphi'}$ can.
To show these equivalences, we prove that the equivalences classes of $\simeq$ correspond exactly to the different values of the Euler characteristic.
We think that this last combinatorial result can be of independent interest.
We then obtain the equivalences (a-b-c) by observing that our transformation always preserves the corresponding properties.
As a bonus, the first equivalence (a) also allows us to show \#P-hardness of some $\mathcal{H}$-queries that are not UCQs, thus slightly extending the dichotomy of~\cite{dalvi2012dichotomy} for the $\mathcal{H}$-queries.

\begin{figure}
\centering
\begin{tikzpicture}
	\draw[draw=black,very thick] (0,0) rectangle ++(\linewidth,3);%all H queries
	\draw[draw=black!50!green,very thick,dashed] (0.1,0.1) rectangle ++(3,2.8);%fragmentable = null euler charact \implies have d-Ds
	\draw[draw=black,very thick] (1.5,0.5) rectangle ++(4,2);%monotone H queries = the ones that are UCQs
	\draw[draw=black!50!green,very thick] (1.6,0.6) rectangle ++(1.4,1.8);%the safe H+ querries according to Dalvi Suciu
	\draw[draw=black!20!red,very thick,dashed] (3.2,0.1) rectangle ++(4,2.8);%nonzero euler characteristic achievable by monotone function; we show are hard
	\draw[draw=black!20!red,very thick] (3.3,0.6) rectangle ++(2.1,1.8);%the unsafe H+ queries according to dichotomy
	\draw[draw=black!50!blue,very thick] (1.9,0.2) rectangle ++(0.4,2.6);%the H queries that are inversion free = degenerate
	\draw[draw=black!20!gray,very thick,dotted] (7.3,0.1) rectangle ++(\linewidth-7.4cm,2.8);%euler characteristic not achievable by monotone function; we don't know
\end{tikzpicture}
	\caption{Picture of the situation. Outmost black rectangle: all $\mathcal{H}$-queries.
	Interior black horizontal rectangle: all $\mathcal{H}^+$-queries (UCQs).
	Solid green rectangle: all safe $\mathcal{H}^+$-queries according to the dichotomy of~\cite{dalvi2012dichotomy},
	i.e., whose CNF lattice	have zero Möbius value.
	Solid red rectangle: all unsafe $\mathcal{H}^+$-queries.
	Blue vertical rectangle: all the $\mathcal{H}$-queries having OBDDs constructible in PTIME
	(upper bound~\cite{fink2016dichotomies}, lower bound~\cite{beame2017exact}).
	These are the inversion-free $\mathcal{H}$-queries and correspond to what we call degenerate Boolean functions.
	Dashed green rectangle (left): all the $\mathcal{H}$-queries to which our technique can apply, implying that they have \mbox{d-Ds} constructible in PTIME.
	These correspond to fragmentable Boolean functions, or equivalently, to those having zero Euler characteristic.
	Dashed red rectangle (right): the $\mathcal{H}$-queries corresponding to Boolean functions having non-zero Euler characteristic and for which we could show \#P-hardness.
	Dotted gray rectangle: all the remaining $\mathcal{H}$-queries corresponding to Boolean functions having non-zero Euler characteristic. We conjecture these to be
	\#P-hard as well.
	Note that each individual region contains an infinite number of queries.
	}
\label{fig:full-picture}
\end{figure}
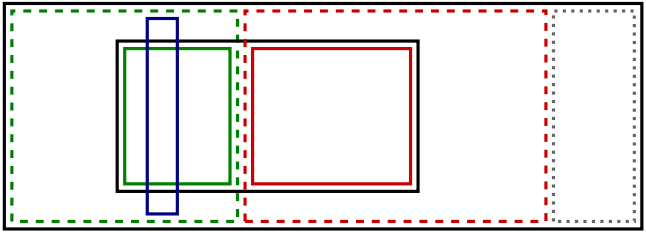

\paragraph{\bf The difference with~\cite{monet2018towards}}
In~\cite{monet2018towards}, we studied with Dan Olteanu the compilation of the $\mathcal{H}$-queries into d-DNNFs.
This preliminary work already contains the idea of covering the satisfying valuations of~$\varphi$ by removing
two adjacent satisfying valuations, through the notion of what we
called a \emph{nice Boolean function} (of which fragmentability is a generalization).
This article contained an experimental study that the proposed approach might work for the safe $\mathcal{H}^+$-queries, but had no proof.
We still do not know if this approach works for all safe $\mathcal{H}^+$-queries,
but we do know that it cannot work for all $\mathcal{H}$-queries with zero Euler characteristic (more on that in Section~\ref{sec:open}).
What was missing to get to a proof for the $\mathcal{H}$-queries is the reformulation with the Euler characteristic
(which allows one to understand the combinatorial meaning of
having a zero Möbius value of the CNF lattice), and the fact that we might also need to \emph{add} adjacent non-satisfying valuations.

\paragraph{\bf Paper structure}

To help the reader, we have depicted the structure of the paper and of the proofs as a DAG in Appendix~\ref{apx:proof-structure}.
We start in Section~\ref{sec:preliminaries} with short preliminaries.
In Section~\ref{sec:hk} we define the $\mathcal{H}$-queries, review what is known about them and reformulate the safety criterion
for the $\mathcal{H}^+$-queries in terms of the Euler characteristic.
We continue in Section~\ref{sec:fragmentability} by introducing our notion of fragmentable Boolean function and we prove
that if~$\varphi$ is fragmentable then $Q_\varphi$ has \mbox{d-Ds}.
In Section~\ref{sec:upper_bound} we present our transformation and prove that if~$\varphi$ has zero Euler characteristic then~$\varphi$ is fragmentable, which implies our main result.
We analyze in Section~\ref{sec:hardness} what happens in the case of a non-zero Euler characteristic.
Last, we expose in Section~\ref{sec:open} open questions, justify the definition of our transformation, and discuss the applicability of
our technique to a broader class of queries than the $\mathcal{H}$-queries.
We conclude in Section~\ref{sec:conclusion}.

\section{Preliminaries}
  \label{sec:preliminaries}
  \paragraph{\bf Boolean functions.}

A (Boolean) \emph{valuation} of a set $V$ is a subset of~$V$.
We write $2^V$ the set of all valuations of $V$.
For a Boolean valuation $\nu$ and variable $l$ of $V$, we write $\nu^{(l)} \defeq \begin{cases} 
	\nu \cup \{l\} & \mbox{if } l \notin \nu \\
	\nu \setminus \{l\} & \mbox{if } l \in \nu \end{cases}$, i.e., $\nu^{(l)}$ is $\nu$ except that the membership of $l$ has been flipped.
A \emph{Boolean function}~$\varphi$ on $V$ is a mapping
$\varphi:2^V\to\{\false,\true\}$ that associates
to each valuation $\nu$ a value~$\varphi(\nu)$ in $\{\false, \true\}$.
A valuation $\nu$ is \emph{satisfying} when $\varphi(\nu) = \true$, also written \mbox{$\nu \models \varphi$}.
We write $\sat(\varphi)$ the set of satisfying valuations of~$\varphi$, and $\#\varphi$ the number of satisfying valuations of~$\varphi$.
We write~$\bot$ (resp., $\top$) the Boolean function that maps every valuation to~$\false$ (resp., $\true$).
A Boolean function~$\varphi$ is \emph{monotone} if $\varphi(\nu) \implies \varphi(\nu')$ for every valuations $\nu, \nu'$ such that $\nu \subseteq \nu'$. 
A monotone Boolean function can always be represented as a DNF of the form~$C_0 \lor \ldots \lor C_n$,
where we see each clause simply as the set of variables that it contains. Moreover, there is a unique \emph{minimized} DNF representing~$\varphi$, where no clause is a subset of another; we denote it by~$\varphi_\dnf$. We will similarly consider
the unique minimized CNF representation~$\varphi_\cnf$ of a monotone Boolean function~$\varphi$.\footnote{Note that the uniqueness of~$\varphi_\cnf$ is less obvious than for~$\varphi_\dnf$, see for instance~\cite{cstheory_unicity_cnfs}.}
Two Boolean functions $\varphi,\varphi'$ are \emph{disjoint} when $\varphi \land \varphi' = \bot$, or in other words, when $\sat(\varphi \land \varphi') = \emptyset$.

\begin{definition}
\label{def:degenerate}
	Let~$\varphi$ be a Boolean function on $V$.
	We say that~$\varphi$ \emph{depends on variable $l \in V$} when there exists a valuation $\nu \subseteq V$
	such that $\varphi(\nu) \neq \varphi(\nu^{(l)})$.
	We write $\dep(\varphi) \subseteq V$ for the set of variables on which~$\varphi$ depends.
	We call~$\varphi$ \emph{nondegenerate} if $\dep(\varphi) = V$, and \emph{degenerate} otherwise.
\end{definition}

\begin{definition}[{See~\cite{roune2013complexity,stanley2011enumerative}}]
\label{def:euler}
	Let~$\varphi$ be a Boolean function. The \emph{Euler characteristic of~$\varphi$}, denoted $\eul(\varphi)$, is $\sum_{\nu \models \varphi} (-1)^{|\nu|}$.
\end{definition}

\paragraph{\bf Tuple-independent databases.}

A \emph{tuple-independent (TID)\linebreak
database} is a pair $(D,\pi)$ consisting of a relational instance $D$ and a function $\pi$
mapping each tuple $t \in D$ to a rational probability $\pi(t) \in [0;1]$.
A TID instance $(D, \pi)$ defines a probability distribution~$\Pr$ on $D' \subseteq D$, where
$\Pr(D') \defeq \prod_{t \in D'} \pi(t) \times \prod_{t \in D \backslash D'} (1 - \pi(t))$.
Given a Boolean query $Q$, the \emph{probabilistic query evaluation problem for $Q$} ($\pqe(Q)$) asks, given as input a
TID instance $(D, \pi)$, the probability that $Q$ is satisfied in the distribution $\Pr$.
That is, formally, $\Pr(Q, (D, \pi)) \defeq \sum_{D' \subseteq D\text{~s.t.~}D' \models Q} \Pr(D')$. 
Dalvi and Suciu~\cite{dalvi2012dichotomy} have shown a dichotomy result on UCQs for $\pqe$: either $Q$ is \emph{safe}
and $\pqe(Q)$ is PTIME, or $Q$ is not safe and $\pqe(Q)$ is \#P-hard. 
Due to space constraints, we must refer to~\cite{dalvi2012dichotomy} for a presentation of their algorithm to compute the probability
of a safe query, though it is not necessary to understand the current paper: we will reproduce the relevant parts in the next section.
We write $\pqe(Q) \leq_{\mathrm{T}} \pqe(Q')$ when $\pqe(Q)$ reduces in PTIME (under Turing reductions) to $\pqe(Q')$.
We write $\pqe(Q) \equiv_{\mathrm{T}} \pqe(Q')$ when we have both $\pqe(Q) \leq_{\mathrm{T}} \pqe(Q')$ and $\pqe(Q') \leq_{\mathrm{T}} \pqe(Q)$.
\paragraph{\bf Lineages.}

The \emph{lineage} of a Boolean query $Q$ over~$D$ is a Boolean
function $\lin(Q, D)$ on the tuples of~$D$
mapping each sub-database $D' \subseteq D$ of $D$
to~$\true$ or~$\false$ depending on whether $D'$ satisfies~$Q$ or not.
It is clear that, by definition, for any probability mapping $\pi:D \to [0;1]$ we have $\Pr(Q, (D, \pi)) = \Pr(\lin(Q,D),\pi)$, where the later is defined just 
as expected (see Definition~\ref{def:prob-function} in Appendix~\ref{apx:mobius-euler} for the formal definition).
This connection was first observed in~\cite{green2006models}.

\paragraph{\bf Knowledge compilation formalisms.}
The lineage, being a Boolean function, can be represented with any formalism that represents Boolean functions (Boolean formulas, BDDs, Boolean circuits, etc).
In the intensional query evaluation context, the crucial idea is
to use a formalism that allows tractable probability computation.
In this work we will specifically focus on \emph{deterministic decomposable circuits} (\mbox{d-Ds}) and 
\emph{deterministic decomposable normal forms}~\cite{darwiche2001tractability} (d-DNNFs).
Let $C$ be a Boolean circuit (featuring $\land$, $\lor$, $\lnot$, and variable gates). 
An $\land$-gate $g$ of $C$ is \emph{decomposable} if for every two input gates
$g_1\neq g_2$ of $g$ we have $\VARS(g_1) \cap \VARS(g_2) = \emptyset$, where $\VARS(g)$ denotes the set of variable gates that have a directed path to~$g$ in~$C$.
We call $C$ \emph{decomposable} if each $\land$-gate is.
An $\lor$-gate $g$ of $C$ is \emph{deterministic} if
for every pair $g_1\neq g_2$ of input gates of~$g$, the Boolean functions captured by $g_1$ and $g_2$ are disjoint.
We call $C$ \emph{deterministic} if each $\lor$-gate is.
A \emph{negation normal form} (NNF) is a circuit in which the inputs of $\lnot$-gates are always variable gates.
Probability computation is in linear time for \mbox{d-Ds} (hence, for d-DNNFs):
to compute the probability of a d-D, compute by a bottom-up pass the probability of each gate,
where $\land$ gates are evaluated using $\times$, $\lor$ gates using $+$, and $\lnot$ gates using $1-x$.

For a formalism $\mathcal{C}$ of representation of Boolean functions, we denote by $\mathcal{C}(\mathrm{PTIME})$ the class of
all Boolean queries~$Q$ such that, given as input an arbitrary database~$D$, we can compute in polynomial time (in data complexity) an element of~$\mathcal{C}$ representing~$\lin(Q,D)$.
Similarly, we write $\mathcal{C}(\mathrm{PSIZE})$ for the class of Boolean queries~$Q$ such that, for any database~$D$, there exists an element of~$\mathcal{C}$ representing~$\lin(Q,D)$ whose size is polynomial in that of~$D$ (but this element is not necessarily computable in PTIME).

\paragraph{\bf The Möbius function.}
The characterization of the safe $\mathcal{H}^+$-queries  by~\cite{dalvi2012dichotomy} uses the Möbius function of a poset, which we define here.
Let $P = (A,\leq)$ be a finite poset.
The Möbius function~\cite{stanley2011enumerative} $\mu_P: A \times A \to \mathbb{Z}$ of $P$ is defined on pairs $(u,v)$ with $u \leq v$ by $\mu_P(u,u) \defeq 1$ and 
$\mu_P(u,v) \defeq - \sum\limits_{u < w \leq v} \mu_P(w,v)$ for $u < v$. It is used to express the \emph{Möbius inversion formula},
which is a generalization for posets of the inclusion--exclusion principle (see Proposition~\ref{prp:mobius-inversion} in Appendix~\ref{apx:mobius-euler}).

\section{The $\mathcal{H}$-queries}
  \label{sec:hk}
  In this section we define the $\mathcal{H}$-queries and review what is known about them, before reformulating the safety criterion for $\mathcal{H}^+$-queries.

\subsection{Reviewing what is known}
\label{sec:review}

The building blocks of these queries are the conjunctive queries $h_{k,i}$, which were first defined in the work of Dalvi and Suciu~\cite{dalvi2012dichotomy} to show the hardness of UCQs that are not safe:

\begin{definition}
\label{def:hk}
	Let $k \in \mathbb{N}_{\geq 1}$, and let~$R,S_1,\ldots,S_k,T$ be pairwise distinct relational predicates, with~$R$ and~$T$ being unary, and~$S_i$ for~$0 \leq i \leq k$ being binary. The queries $h_{k,i}$ for $0 \leq i \leq k$ are defined by:
\begin{itemize}
	\item $h_{k,0} = \exists x \exists y~R(x) \land S_1(x,y)$;
	\item $h_{k,i} = \exists x \exists y~S_i(x,y) \land S_{i+1}(x,y)$ for $1 \leq i < k$;
	\item $h_{k,k} = \exists x \exists y~S_k(x,y) \land T(y)$.
\end{itemize}
\end{definition}

As in~\cite{beame2017exact}, we define the $\mathcal{H}_k$-queries to be Boolean combinations of the queries $h_{k,i}$:

\begin{definition}
	Let $k \in \mathbb{N}_{\geq 1}$ and~$\varphi$ be a Boolean function on variables
	$V=\{0,\ldots,k\}$.
	We define the query $Q_\varphi$ to be the query represented by the first order formula
	$\varphi[0 \mapsto h_{k,0}, \ldots, k \mapsto h_{k,k}]$, i.e.,~$\varphi$ where we substituted each variable $i \in V$ by the formula $h_{k,i}$.
\end{definition}

The query class $\mathcal{H}_k$ (resp., $\mathcal{H}^+_k$) is then the set of queries $Q_\varphi$ when~$\varphi$ ranges over all 
Boolean functions (resp., monotone Boolean functions) on variables $\{0,\ldots,k\}$.
We finally define $\mathcal{H}$ (resp., $\mathcal{H}^+$) to be $\bigcup\limits_{k=1}^\infty \mathcal{H}_k$ (resp., $\bigcup\limits_{k=1}^\infty \mathcal{H}^+_k$).
Observe that the queries in $\mathcal{H}^+$ are in particular (equivalent to some) UCQs.

\begin{example}
\label{expl:q9}
	Let $k=3$, and $\varphi_9$ be the monotone Boolean function $(2 \vee 3) \wedge (0 \vee 3) \wedge (1 \vee 3) \wedge (0 \vee 1 \vee 2)$. Then $Q_{\varphi_9}$ 
	represents the query $(h_{32} \vee h_{33}) \wedge (h_{30} \vee h_{33}) \wedge (h_{31} \vee h_{33}) \wedge (h_{30} \vee h_{31} \vee h_{32}) \in \mathcal{H}^+_3$.
	This query was introduced in~\cite{dalvi2012dichotomy}, where it is called $q_9$.
	It is the simplest safe $\mathcal{H}^+$-query
	for which Dalvi and Suciu's algorithm requires the Möbius inversion formula.
\end{example}

At this point, it is important to warn the reader not to confuse the function~$\varphi$, whose purpose is to define $Q_\varphi$, and the lineage of $Q_\varphi$ on a database.
While they are both Boolean functions, the similarity stops here: we use~$\varphi$ simply to alleviate the notations.
Also, from now on and until the end, we fix $k\in \mathbb{N}_{\geq 1}$ and $V=\{0,\ldots,k\}$, and we will never use the symbols $k$ and $V$ for anything else.

To study the $\mathcal{H}_k$-queries, let us first listen to what the dichotomy of~\cite{dalvi2012dichotomy} has to say about them. 
We shall temporarily restrict our attention to monotone functions, i.e., to queries in $\mathcal{H}_k^+$, because the dichotomy theorem applies only to UCQs.
We need to define the \emph{CNF lattice} of~$\varphi$:

\begin{definition}
\label{def:lattice}
	Let~$\varphi$ be a monotone Boolean function represented by its (unique) minimized CNF
	$\varphi_\cnf = C_0 \land \ldots \land C_n$ (remember that we see each clause simply as the set of variables that it contains). 
	For~$\mathbf{s} \subseteq \{0,\ldots,n\}$, we define $d_\mathbf{s} \defeq \bigcup\limits_{i \in \mathbf{s}} C_i$. 
	Note that $d_\emptyset$ is $\emptyset$, and that it is possible to have $d_\mathbf{s} = d_\mathbf{s'}$ for
	$\mathbf{s} \neq \mathbf{s'}$.
	We define the poset $L_\cnf^\varphi = (A,\leq)$, where 
	$A$ is $\{ d_\mathbf{s} \mid \mathbf{s} \subseteq \{0,\ldots,n\} \}$, and where $\leq$ is reversed set inclusion\footnote{We could have defined it with $\leq$ being set inclusion, this is just a matter of taste.
	We chose to use reversed set inclusion to fit with related work.}.
	One can check that $L_\cnf^\varphi$ is a (finite) lattice, but this fact is not important for this paper.
	In particular, $L_\cnf^\varphi$ has a greatest element $\hat{1}$, which is~$\emptyset$, and has a least element~$\hat{0}$, which is~$\dep(\varphi)$.
\end{definition}

The dichotomy theorem then tells us the following:

\begin{proposition}[({\cite{dalvi2012dichotomy}, see also~\cite[Section 3.3]{beame2017exact}})]
\label{prp:complexity_H-queries}
	Let~$\varphi$ be monotone. If~$\varphi$ is degenerate then $\pqe(Q_\varphi)$ is PTIME.
	If~$\varphi$ is nondegenerate, let $L_\cnf^{\varphi}$ be the CNF lattice of~$\varphi$, and $\mu_\cnf$ be the Möbius function on $L_\cnf^{\varphi}$.
	Then, if $\mu_\cnf(\hat{0},\hat{1})=0$ then $\pqe(Q_\varphi)$ is PTIME, otherwise it is \#P-hard.
\end{proposition}

\begin{example}
\label{expl:lattice}
	Consider $\varphi_9$ from Example~\ref{expl:q9}. Clearly $\varphi_9$ is nondegenerate.
	The Hasse diagram of its CNF lattice is shown in Figure~\ref{fig:lattice}, with the value $\mu_\cnf(n,\hat{1})$ for each node $n$.
	We have $\mu_\cnf(\hat{0},\hat{1})=\mu_\cnf(\{0,1,2,3\},\emptyset)=0$, hence $\pqe(Q_{\varphi_9})$ is PTIME.
\end{example}

The solid green and red rectangles in Figure~\ref{fig:full-picture} represent the tractable and \#P-hard $\mathcal{H}^+$-queries.

Having reviewed the tractability of the $\mathcal{H}^+$-queries, we now do the same concerning the compilation of $\mathcal{H}$-queries to knowledge compilation formalisms.
First, when~$\varphi$ is degenerate, then $Q_\varphi \in \obddptime$ (hence $Q_\varphi \in \ddptime$):

\begin{toappendix}
	\subsection{\mbox{d-Ds} for degenerate $\mathcal{H}$-queries}
	\label{apx:fink2016dichotomies}
	In this section we explain how to build \mbox{d-Ds} in polynomial time for $\mathcal{H}$-queries $Q_\varphi$ such that~$\varphi$ is degenerate.
	Remember that we have:
\end{toappendix}

\begin{propositionrep}[(Implied by Lemma 3.8 of~\cite{fink2016dichotomies})]
\label{prp:if-degenerate-then-obddptime}
	If~$\varphi$ (not necessarily monotone) is degenerate, then
	$Q_\varphi \in \obddptime$.
\end{propositionrep}
\begin{toappendix}
	We only explain how it works for \mbox{d-Ds}, since this is what we need (but we will still use OBDDs in the proof, and we assume here that the reader is familiar with
their definition).
For a valuation $\nu \subseteq V$, we write $\varphi_\nu$ the Boolean function defined by $\sat(\varphi_\nu) \defeq \{\nu\}$.
Since~$\varphi$ is degenerate, let $l \in V$ be a variable upon which it does not depend. We can then write~$\varphi$ as
$\varphi = \bigvee_{\nu \models \varphi, l \notin \nu} (\varphi_\nu \lor \varphi_{\nu^{(l)}})$.
The outermost~$\lor$ being disjoint, we only need to explain how to build \mbox{d-Ds} for the queries $Q_{\varphi_\nu \lor \varphi_{\nu^{(l)}}}$.
The query $Q_{\varphi_\nu \lor \varphi_{\nu^{(l)}}}$ can be written as 
$Q^L \land Q^R$, where $Q^L$ is
$\bigwedge_{i \in P_L} h_{k,i} \land \bigwedge_{i \in N_L} \lnot h_{k,i}$ for some
subsets $P_L,N_L \subseteq \{0,\ldots,l-1\}$ with $P_L \cap N_L = \emptyset$, and similarly
$Q^R$ is
$\bigwedge_{i \in P_R} h_{k,i} \land \bigwedge_{i \in N_R} \lnot h_{k,i}$ for some
subsets $P_R,N_R \subseteq \{l+1,\ldots,k\}$ with $P_R \cap N_R = \emptyset$.
Since the Boolean queries $Q^L$ and $Q^R$ do not share any relational predicate, by decomposability it is enough to show how to build \mbox{d-Ds} for these two, separately.
We only show it for $Q^L$ since $Q^R$ works similarly.
To do this, we will build OBDDs $O_i$ for the queries $h_{k,i}$, $i \in \{0,\ldots,l-1\}$, under the same variable order $\Pi_L$ (recall that the variables of the lineage
are the tuples of the database).
This will be enough, because we can then use standard techniques~\cite{wegener2004bdds} to combine these OBDDs and obtain an OBDD for $Q^L$.
Let $\{a_1,\ldots,a_n\}$ be the domain of the database $D$.
The variable order $\Pi_L$ that we consider is $\Pi_L \defeq \Pi_L(1) \ldots \Pi_L(n)$, where
$\Pi_L(i) \defeq R(a_i), S_1(a_i,a_1), \ldots, S_{l-1}(a_i,a_1), S_1(a_i,a_2), \ldots, S_{l-1}(a_i,a_2), \ldots, S_1(a_i,a_n), \ldots, S_{l-1}(a_i,a_n)$.
Under this variable order, it is easy to see that we can build in polynomial time OBDDs $O_i$ for the queries $h_{k,i}$, $i \in \{0,\ldots,l-1\}$, concluding the proof.

\end{toappendix}
In this work we will only need the fact that $Q_\varphi \in \ddptime$ when~$\varphi$ is degenerate.
Since our results depend on Proposition~\ref{prp:if-degenerate-then-obddptime} and in order to be self-contained,
we reproduce its proof in Appendix~\ref{apx:fink2016dichotomies}, but it can be safely skipped: the reader can see it as a black box.

Proposition~\ref{prp:if-degenerate-then-obddptime} is in sharp contrast to when~$\varphi$ is nondegenerate.
Indeed in that case
 the authors of \cite{beame2017exact} show an exponential lower bound on the size of what they call \emph{Decomposable Logic Decision Diagrams} (DLDDs) for $Q_\varphi$.
A DLDD is a d-D in which the determinism of $\lor$-gates
is restricted to simply choosing the value of a variable.
That is, each $\lor$-gate is of the form $(v \land g) \lor (\lnot v \land g')$ for some variable~$v$. 
An OBDD being in particular a DLDD, we have represented in Figure~\ref{fig:full-picture} \emph{all} the $\mathcal{H}$-queries in $\obddptime$ by the vertical blue rectangle (and we even know that those outside this rectangle are not even in $\obddpsize$, thanks to this lower bound). 

When~$\varphi$ is monotone and nondegenerate, another independent exponential lower bound~\cite{bova2017circuit} tells us than we cannot impose \emph{structuredness} either (i.e., use d-SDNNFs).
These two lower bounds mean that to build tractable lineages for the $\mathcal{H}$-queries, one cannot restrict the expressivity of determinism too much, nor restrict the structure in which the variables appear.
One question is then: do the safe nondegenerate $\mathcal{H}^+$-queries have polynomial sized (and computable in PTIME?) d-DNNFs (or \mbox{d-Ds})?

\begin{figure}
\centering
	\begin{scaletikzpicturetowidth}{\linewidth}
	\begin{tikzpicture}[scale=\tikzscale]
	\usetikzlibrary{shapes}
	\tikzset{nodestyle/.style={draw,rectangle}}
	\input{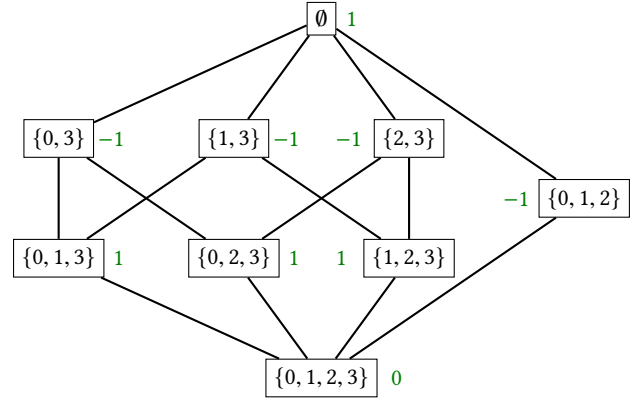}
	\end{tikzpicture}
	\end{scaletikzpicturetowidth}
	\caption{Hasse diagram of $L_\cnf^{\varphi_9}$. The green value besides each node $n$ is $\mu_\cnf(n,\emptyset)$,
	and can be computed top-down following the definition.}
\label{fig:lattice}
\end{figure}

\subsection{Reformulation of safety for the $\mathcal{H}^+$-queries}
\label{sec:mobius-euler}

With this section starts the presentation of our technical contributions.
Here, we reformulate the safety criterion of~\cite{dalvi2012dichotomy} in terms of the Euler characteristic (recall Definition~\ref{def:euler}).
As a by-product, we also show that for the $\mathcal{H}^+$-queries, using the CNF lattice or the \emph{DNF lattice}\footnote{Defined similarly as the CNF lattice but by starting from the DNF minimized expression of~$\varphi$ instead of the CNF expression.} is equivalent (that question was left open in~\cite{monet2018towards}). We show the following (remember that $k\in\mathbb{N}_{\geq 1}$ and $V=\{0,\ldots,k\}$ are fixed):

\begin{toappendix}
	\subsection{Equivalence Möbius--Euler}
	\label{apx:mobius-euler}
	In this section we prove Lemma~\ref{lem:big-coeff}.
	We start by recalling its statement (remember that $k \in \mathbb{N}_{\geq 1}$ and $V=\{0,\ldots,k\}$ are fixed):
\end{toappendix}

\begin{lemmarep}
\label{lem:big-coeff}
	Let~$\varphi$ be a nondegenerate monotone Boolean function on~$V$. Then we have
	$\eul(\varphi) = \mu_\cnf(\hat{0},\hat{1}) = (-1)^k \mu_\dnf(\hat{0},\hat{1})$,
	where $\mu_\cnf$ (resp., $\mu_\dnf$) is the Möbius function of $L^\varphi_\cnf$ (resp., $L^\varphi_\dnf$).
\end{lemmarep}

\begin{proofsketch}
	The idea is to use Möbius's inversion formula on the CNF and DNF lattices to obtain three different expressions of a variant of
	the characteristic polynomial of~$\varphi$~\cite{nisan1994degree}, and then to observe that
	the leading coefficients are equal to the targeted terms.
	The full proof can be found in Appendix~\ref{apx:mobius-euler}.
\end{proofsketch}

\begin{toappendix}
	The proof contains three ingredients. The first one is the Möbius inversion formula:

\begin{proposition}[({See \cite[Proposition 3.7.1]{stanley2011enumerative}})]
\label{prp:mobius-inversion}
	Let $P$ be a finite poset, and let $f,g: P \to \mathbb{R}$. Then the following are equivalent:
	\begin{itemize}
		\item $g(x) = \sum_{u \leq x} f(u)$ for all $x \in P$;
		\item $f(x) = \sum_{u \leq x} \mu(u,x) g(u)$ for all $x \in P$.
	\end{itemize}
\end{proposition}

The second one is the notion of probability of a Boolean function, which we introduce formally here:

\begin{definition}
	\label{def:prob-function}
A \emph{probability assignment $\pi$} is a mapping from $V$ to $[0;1]$.
	Given a valuation $\nu \subseteq V$ and a probability assignment $\pi$, the \emph{probability $\pi(\nu)$ of $\nu$ under $\pi$} is defined as
\[ \pi(\nu) \defeq \left( \prod_{x \in \nu} \pi(x) \right) \left( \prod_{x \in V \setminus \nu} (1 - \pi(x)) \right).\]
	Given a Boolean function~$\varphi$ and a probability assignment $\pi$, \emph{the probability $\Pr(\varphi,\pi)$ of~$\varphi$ under $\pi$} is then naturally defined as the total probability mass under $\pi$ of the valuations that satisfy~$\varphi$, that is
	$\Pr(\varphi,\pi) \defeq \sum_{\nu \models \varphi} \pi(\nu)$.
\end{definition}

We will be using specific probability assignments:

\begin{definition}
	For $t \in [0;1]$, let $\pi_t$ denote the probability assignment that maps every variable to $t$.
\end{definition}

The third ingredient is a univariate variant of a \emph{characteristic polynomial} of~$\varphi$~\cite{nisan1994degree}, of which we will give three different expressions:

\begin{definition}
\label{def:polynomials}
	Let~$\varphi$ be a nondegenerate monotone Boolean function, written as $\bigwedge_{0 \leq i \leq n} C_i$ in CNF and as $\bigvee_{0 \leq i \leq m} C'_i$ in DNF.
	Consider the CNF and DNF lattices of~$\varphi$, $L_\cnf^\varphi$ and $L_\dnf^\varphi$. 
	For $\mathbf{s} \subseteq \{0,\ldots,n\}$ or $\mathbf{s} \subseteq \{0,\ldots,m\}$ we write $d_\mathbf{s} \defeq \bigcup_{i \in \mathbf{s}} C_i$ and $d'_\mathbf{s} \defeq \bigcup_{i \in \mathbf{s}} C'_i$, as well as
	$\alpha_\mathbf{s} \defeq |d_\mathbf{s}|$ and $\alpha'_\mathbf{s} \defeq |d'_\mathbf{s}|$.
	We define the polynomials $P^\varphi$, $P_\dnf^\varphi$, and $P_\cnf^\varphi$ of $\mathbb{R}[t]$ as follows:
	\begin{itemize}
		\item $P^\varphi(t) \defeq \Pr(\varphi,\pi_t)$;
		\item $P_\cnf^\varphi(t) \defeq \mathlarger{\sum}_{x = d_\mathbf{s} \in L_\cnf^\varphi} \mu_\cnf(x,\hat{1}) (1-t)^{\alpha_\mathbf{s}}$;
		\item $P_\dnf^\varphi(t) \defeq 1- \mathlarger{\sum}_{x = d'_\mathbf{s} \in L_\dnf^\varphi} \mu_\dnf(x,\hat{1}) t^{\alpha'_\mathbf{s}}$.
	\end{itemize}
\end{definition}

We will show that these polynomials are equal:

\begin{lemma}
\label{lem:polynomials}
	Let~$\varphi$ be a nondegenerate monotone Boolean function.
	Then for all $t \in \mathbb{R}$, we have that
	$P^\varphi(t) = P_\cnf^\varphi(t) = P_\dnf^\varphi(t)$.
\end{lemma}

This will imply Lemma~\ref{lem:big-coeff}: indeed, observe that the coefficient of $t^{k+1}$ is $\sum_{\nu \models \varphi} (-1)^{k+1-|\nu|} = (-1)^{k+1} \sum_{\nu \models \varphi} (-1)^{|\nu|}$ in $P^\varphi(t)$, and is $(-1)^{k+1}\mu_\cnf(\hat{0},\hat{1})$ in $P_\cnf^\varphi(t)$, and is $- \mu_\dnf(\hat{0},\hat{1})$ in $P_\dnf^\varphi(t)$.
Since $P^\varphi(t)$ and $P_\cnf^\varphi$ and $P_\dnf^\varphi$ are the same polynomials, these coefficients are equal.
So, let us show Lemma~\ref{lem:polynomials}:

\begin{proof}[Proof of Lemma~\ref{lem:polynomials}]
	Clearly, it is enough to show that these polynomials are equal on $[0;1]$.
	We use Proposition~\ref{prp:mobius-inversion} on $L_\cnf^\varphi$ and on $L_\dnf^\varphi$ to compute $P^\varphi(t)$. 
	We start with $L_\cnf^\varphi$.
	Define the functions $f$ and $g$, from $L_\cnf^\varphi$ to $\mathbb{R}$, as follows: let $d_\mathbf{s} \in L_\cnf^\varphi$, then:
	\begin{enumerate}
		\item $f(d_\mathbf{s}) \defeq \Pr\left( (\lnot \bigvee_{i \in \mathbf{s}} C_i) \land \bigwedge_{i \in V \setminus \mathbf{s}} C_i, \pi_t \right)$; in other words the total probability mass of the valuations that do not satisfy any of the (disjunctive) clauses $C_i$ for $i \in \mathbf{s}$ but satisfy all other clauses $C_i$.
		\item $g(d_\mathbf{s}) \defeq \Pr \left( \lnot \bigvee_{i \in \mathbf{s}} C_i, \pi_t \right)$; in other words the total probability mass of the valuations that do not satisfy any of the clauses $C_i$ for $i \in \mathbf{s}$.
	\end{enumerate}

	We clearly have $g(x) = \sum_{u \leq x} f(u)$ for all $x \in L_\cnf^\varphi$, hence by Proposition~\ref{prp:mobius-inversion} we have
	$f(x) = \sum_{u \leq x} \mu_\cnf(u,x) g(u)$.
	Moreover, for $x = d_\mathbf{s} \in L_\cnf^\varphi$, we have that $g(x) = (1 - t)^{\alpha_\mathbf{s}}$.
	Now, since $P^\varphi(t) = f(\hat{1})$, we indeed have that $P^\varphi(t) = P_\cnf^\varphi(t)$.

	Let us now have a look at $L_\dnf^\varphi$.
	The reasoning is similar, but we include it for completeness.
	For $x = d'_\mathbf{s} \in L_\dnf^\varphi$ define:
	\begin{enumerate}
		\item $f(d'_\mathbf{s}) \defeq \Pr \left( ( \bigwedge_{i \in \mathbf{s}} C'_i) \land \lnot \bigvee_{i \in V \setminus \mathbf{s}} C'_i, \pi_t \right)$; in other words the total probability mass of the valuations that satisfy all the (conjunctive) clauses $C'_i$ for $i \in \mathbf{s}$ but none of the other clauses~$C'_i$.
		\item $g(d'_\mathbf{s}) \defeq \Pr \left(  \bigwedge_{i \in \mathbf{s}} C'_i, \pi_t \right)$; in other words the total probability mass of the valuations that satisfy all the clauses $C'_i$ for $i \in \mathbf{s}$.
	\end{enumerate}
	This time, we have $g(x) = \sum_{u \leq x} f(u)$
	and $f(x) = \sum_{u \leq x} \mu_\dnf(u,x) g(u)$
	and $g(x) = t^{\alpha_\mathbf{s}}$
	for all $x = d'_\mathbf{s}\in L_\dnf^\varphi$. Combining with $P^\varphi(t) = 1 - f(\hat{1})$ we obtain
	that $P^\varphi(t) = P_\dnf^\varphi(t)$.
	This finishes the proof.
\end{proof}

\end{toappendix}

A result that looks very similar to Lemma~\ref{lem:big-coeff} is \emph{Philip Hall's theorem}
(see~\cite[Proposition 3.8.6]{stanley2011enumerative} and text above), relating the
Möbius value $\mu_P(\hat{0},\hat{1})$ of a poset $P$ with the Euler characteristic of a certain
\emph{abstract simplicial complex}\footnote{A synonym for the negation of a monotone Boolean function.} defined from $P$.
However, we did not see a direct relation between this result and our lemma.

Combining Lemma~\ref{lem:big-coeff} with Proposition~\ref{prp:complexity_H-queries}, and with the observation that any degenerate function~$\varphi$ has 
$\eul(\varphi)=0$, we obtain the following simple characterization of the safe $\mathcal{H}^+$-queries:

\begin{corollary}
\label{cor:rephrase}
	Let~$\varphi$ be monotone. If $\eul(\varphi)=0$ then $\pqe(Q_\varphi)$ is PTIME, otherwise it is \#P-hard.
\end{corollary}

Using the Euler characteristic instead of the Möbius function to characterize the safe $\mathcal{H}^+$-queries has two advantages.
First, $\eul(\varphi)$ is conceptually much simpler to grasp than $\mu_\cnf(\hat{0},\hat{1})$.\footnote{For instance, a satisfactory consequence is that we can easily express the number of (not necessarily monotone) functions with $\eul(\varphi)=0$: this is 
$\sum_{j=0}^{2^k} \binom{2^k}{j}^2$
.}
The second one is that $\eul(\varphi)$ is always defined, whereas it is not entirely clear how to define the CNF lattice when~$\varphi$ is not monotone.

While the proof of Lemma~\ref{lem:big-coeff} is not very challenging,
the connection does not seem to have been made in the literature so far.

\section{Fragmentable Boolean Functions}
  \label{sec:fragmentability}
  We now introduce our notion of fragmentable Boolean function and show that whenever~$\varphi$ is fragmentable, then $Q_\varphi \in \ddptime$.
We also show that if $\eul(\varphi) \neq 0$ then~$\varphi$ cannot be fragmentable.

\subsection{Definition}

We first need to define what we call an \emph{$\lnot$-$\lor$-template}.

\begin{definition}
\label{def:template}
	An \emph{$\lnot$-$\lor$-template} is a Boolean circuit whose internal nodes (i.e., those that are not a leaf) are either $\lnot$- or $\lor$-gates.
	We call \emph{holes} the variables (i.e., the leaves) of~$T$.
	Let $l_0,\ldots,l_n$ be the holes of $T$, and let $\varphi_0,\ldots,\varphi_n$ be Boolean functions.
	Then $T[\varphi_0,\ldots,\varphi_n]$ represents the Boolean function obtained by substituing each hole $l_i$
	by~$\varphi_i$, and then seeing the result as a “hybrid” Boolean circuit\footnote{“Hybrid” because the leaves of the circuit are Boolean functions, i.e., we do not care how these are represented concretely.}.
	We say that $T[\varphi_0,\ldots,\varphi_n]$ is \emph{deterministic} when every $\lor$-gate \emph{in the template} is deterministic.
\end{definition}

Observe that it can be the case that $T[\varphi_0,\ldots,\varphi_n]$ is deterministic while $T$ itself is not.
For a simple example, take $T$ to be the template with two holes $l_0 \lor l_1$,
and take $\varphi_0 \defeq x$ and $\varphi_1 \defeq \lnot x$.
Then $T[\varphi_0,\varphi_1]$ is deterministic but $T$ is not.
Proposition~\ref{prp:if-degenerate-then-obddptime} then suggests the following definition:

\begin{definition}
\label{def:fragmentablility}
	We say that a Boolean function~$\varphi$ is \emph{fragmentable} if there exist a $\lnot$-$\lor$-template $T$ and degenerate
	Boolean functions $\varphi_0,\ldots,\varphi_n$ (with $n+1$ equals the number of holes of $T$)
	such that $T[\varphi_0,\ldots,\varphi_n]$ is deterministic and is equivalent to~$\varphi$.
\end{definition}

Note that in Definition~\ref{def:template} we allowed $\lnot$-$\lor$-templates to consist of a single leaf, which is then also the root.
In particular, this implies that any degenerate Boolean function is fragmentable.

\begin{example}
\label{expl:q9-fragmentable}
	Consider $\varphi_9$ from Example~\ref{expl:q9}.
	Let $T$ be the $\lnot$-$\lor$-template~$T \defeq l_0 \lor l_1 \lor l_2 \lor l_3$,
	and let~$\varphi_0 \defeq 0 \land \lnot 2 \land 3$; $\varphi_1 \defeq \lnot 1 \land 2 \land 3$;  $\varphi_2 \defeq \lnot 0 \land 1 \land 3$; and
	$\varphi_3 \defeq 0 \land 1 \land 2$, which are all degenerate.
	One can easily see that $T[\varphi_0,\varphi_1,\varphi_2,\varphi_3]$ is deterministic (because any two $\varphi_i,\varphi_j$ with $i\neq j$ are disjoint)
	and (less easily) that $T[\varphi_0,\varphi_1,\varphi_2,\varphi_3]$, i.e., $\varphi_0 \lor \varphi_1 \lor \varphi_2 \lor \varphi_3$, is equivalent to~$\varphi_9$.
	Therefore, $\varphi_9$ is fragmentable. Observe the we did not use $\lnot$-gates in the template.
\end{example}

\subsection{First properties}

The main property of fragmentability is that it implies having \mbox{d-Ds} constructible in polynomial time.
This is actually its sole purpose.

\begin{proposition}
\label{prp:if-fragmentable-then-ddptime}
	If~$\varphi$ is fragmentable then $Q_\varphi \in \ddptime$.
\end{proposition}
\begin{proof}
	Let $T$ be a $\lnot$-$\lor$-template and $\varphi_0,\ldots,\varphi_n$ be degenerate Boolean functions such that $T[\varphi_0,\ldots,\varphi_n]$
	if deterministic and equivalent to~$\varphi$, and let $D$ be an arbitrary database.
	Since $\varphi_0,\ldots\varphi_n$ are degenerate, by Proposition~\ref{prp:if-degenerate-then-obddptime} we can construct in PTIME
	deterministic decomposable Boolean circuits $C_0,\ldots,C_n$
	capturing the lineages $\lin(Q_{\varphi_i},D)$ of $Q_{\varphi_i}$ on $D$, for $0 \leq i \leq n$.
	Let $C_\varphi$ be the Boolean circuit obtained by plugging the circuits~$C_i$ at the holes of $T$.
	Then one can check that $C_\varphi$ is a d-D and that it captures~$\lin(Q_\varphi,D)$.
\end{proof}

Remember that we are working with data complexity here, so we can assume that we know $T$ and $\varphi_0,\ldots,\varphi_n$ already.
Although we do not know the exact complexity of finding these given~$\varphi$,
the results of the next section will imply (Corollary~\ref{cor:computable}) that this task is computable.

\begin{example}
\label{expl:q9-ddptime}
	Continuing Example~\ref{expl:q9-fragmentable}, since $\varphi_9$ is fragmentable we obtain that $Q_{\varphi_9} \in \ddptime$\footnote{
		We do not claim credit for this fact, as it has been known for at least 4 years by Guy Van den Broeck
		and Dan Olteanu (with whom we initially worked on the problem).}.
\end{example}

Another property of fragmentability is that it implies having zero Euler characteristic.

\begin{proposition}
\label{prp:if-fragmentable-then-euler-null}
	If~$\varphi$ is fragmentable then $\eul(\varphi)=0$.
\end{proposition}
\begin{proof}
	Easily proved by bottom-up induction on $T[\varphi_0,\ldots,\varphi_n]$ and by using the facts that
	(1) $\eul(\varphi)=0$ when~$\varphi$ is degenerate;
	(2)~$\eul(\lnot\varphi)=-\eul(\varphi)$; and
	(3) $\eul(\varphi \lor \varphi') = \eul(\varphi) + \eul(\varphi')$ when~$\varphi$ and $\varphi'$ are disjoint.
\end{proof}

At first glance, the definition of fragmentability seems completely ad-hoc.
Indeed, it could very well be the case that some $\mathcal{H}$-query, say even $\mathcal{H}^+$-query, is in $\ddptime$
by other means than Proposition~\ref{prp:if-fragmentable-then-ddptime}.
Yet, we will show in the next section the surprising fact that the converse of Proposition~\ref{prp:if-fragmentable-then-euler-null} is also true,
implying that fragmentability is the right notion to consider.

\section{Upper Bound}
  \label{sec:upper_bound}
  We dedicate this section to showing that having a zero Euler characteristic implies being fragmentable.
Formally:

\begin{proposition}
\label{prp:if-euler-null-then-fragmentable}
	If $\eul(\varphi)=0$ then~$\varphi$ is fragmentable.
\end{proposition}

Before embarking on its proof, we spell out its consequences.
The first one is that all the $\mathcal{H}$-queries $Q_\varphi$ with $\eul(\varphi)=0$ are in $\ddptime$, thanks to Proposition~\ref{prp:if-fragmentable-then-ddptime}:

\begin{theorem}
\label{thm:if-euler-null-then-ddptime}
	If $\eul(\varphi)=0$ then $Q_\varphi \in \ddptime$.
\end{theorem}

This is, in a sense, the main result of this article.
We have represented in Figure~\ref{fig:full-picture} all the $\mathcal{H}$-queries with zero Euler characteristic by the dashed green rectangle.
Theorem~\ref{thm:if-euler-null-then-ddptime} applies in particular to all the safe $\mathcal{H}^+$-queries, thanks to Corollary~\ref{cor:rephrase}.
So we get:

\begin{corollary}
\label{cor:all-safe-H+-queries-ddptime}
	All safe $\mathcal{H}^+$-queries are in $\ddptime$.
\end{corollary}

Remember that the intensional--extensional conjecture was \linebreak
thought to hold for the safe $\mathcal{H}^+$-queries, because the use of inclusion--exclusion seemed unavoidable.
This surprising result shows that it \emph{is} avoidable, and that we can get away by using only decomposability and determinism.

Less importantly, by combining Proposition~\ref{prp:if-euler-null-then-fragmentable} and Proposition~\ref{prp:if-fragmentable-then-euler-null}, we obtain the following:

\begin{corollary}
\label{cor:euler-null-iff-fragmentable}
	$\varphi$ is fragmentable if and only if $\eul(\varphi)=0$.
\end{corollary}

In the rest of this section, we prove Proposition~\ref{prp:if-euler-null-then-fragmentable}.
The idea is the following.
In Section~\ref{sec:transformation}, we define a notion of transformation between Boolean functions.
This transformation defines an equivalence relation (denoted $\simeq$) between all Boolean functions, and we show that for
a Boolean function~$\varphi$, if $\varphi \simeq \bot$ then~$\varphi$ is fragmentable.
Then, in Section~\ref{sec:if-euler-null-then-equiv-bot} we show that $\varphi \simeq \bot$ when $\eul(\varphi)=0$.

\subsection{The transformation}
\label{sec:transformation}

(We remind once again the reader that the set $V=\{0,\ldots,k\}$ of variables is fixed, and that we only consider Boolean functions with variables $V$.)
We start with the definition of our notion of transformation between Boolean functions:

\begin{definition}
\label{def:transformation}
	Let $\varphi, \varphi'$ be Boolean functions, $\nu \subseteq V$ a valuation and $l \in V$ a variable.
	Then we write $\varphi \rewr{$+(\nu,l)$}{4} \varphi'$ whenever $\nu \not\models \varphi$ and $\nu^{(l)} \not\models \varphi$ and
	$\sat(\varphi') = \sat(\varphi) \cup \{\nu, \nu^{(l)}\}$.
	Similarly, we write $\varphi \rewr{$-(\nu,l)$}{4} \varphi'$ whenever $\varphi' \rewr{$+(\nu,l)$}{4} \varphi$, i.e., when $\nu \models \varphi$ and $\nu^{(l)} \models \varphi$ and
	$\sat(\varphi') = \sat(\varphi) \setminus \{\nu, \nu^{(l)}\}$.
	We will also use the following auxiliary relations:
	\begin{itemize}
		\item we write $\varphi \rewr{$\pm(\nu,l)$}{4} \varphi'$ when $\varphi \rewr{$+(\nu,l)$}{4} \varphi'$ or $\varphi \rewr{$-(\nu,l)$}{4} \varphi'$;
		\item we write $\varphi \rewr{$+$}{4} \varphi'$ when $\varphi \rewr{$+(\nu,l)$}{4} \varphi'$ for some $\nu,l$, and similarly for $\varphi \rewr{$-$}{4} \varphi'$ and
			$\varphi \rewr{$\pm$}{4} \varphi'$;
		\item we write, $\rewr{$+$}{4}^*$, $\rewr{$-$}{4}^*$ and $\rewr{$\pm$}{4}^*$ for the reflexive
			transitive closures thereof.
	\end{itemize}
\end{definition}

It is clear from the definitions that $\rewr{$\pm$}{4}^*$ is symmetric,
so let us write $\simeq$ the induced equivalence relation.
Next, we explain how to visually understand this transformation. 

\begin{definition}
\label{def:induced-subgraph}
	Let $\mathbf{G}_V$ be the undirected graph with node set~$2^V$ and edge set $\{\{\nu,\nu^{(l)}\} \mid \nu \subseteq V, l\in V\}$.
	Let~$\varphi$ be a Boolean function. We define the colored graph $\mathbf{G}_V[\varphi]$ to be the graph $\mathbf{G}_V$ where we have colored every satisfying valuation
	of~$\varphi$.
\end{definition}

\begin{example}
	Consider again the function $\varphi_9$ from Example~\ref{expl:q9} (here $V=\{0,1,2,3\})$
	The graph $\mathbf{G}_V[\varphi_9]$ is shown in Figure~\ref{fig:G_V-q9} (ignore for now the fact that some edges are dashed and green).
\end{example}

\begin{figure}
\centering
	\begin{scaletikzpicturetowidth}{\linewidth}
	\begin{tikzpicture}[scale=\tikzscale]
	\tikzset{nodestyle/.style={draw,rectangle}}
	% Generated with the python script in ../code/generate_tikz_for_posets.py, using the file ../code/functions/q9-poset
% AND THEN MODIFIED BY HAND TO COLOR A PATH IN GREEN TO ILLUSTRATE CHAINSWAPPING LATER
% SO DO NOT DELETE THIS FILE

 ===== DRAWING NODES ====

\node[nodestyle] (emptyset) at (0.0, 0.0) {$\emptyset$};
\node[nodestyle] (0) at (-2.85, 1.3) {$\{0\}$};
\node[nodestyle] (1) at (-0.95, 1.3) {$\{1\}$};
\node[nodestyle] (2) at (0.95, 1.3) {$\{2\}$};
\node[nodestyle] (3) at (2.85, 1.3) {$\{3\}$};
\node[nodestyle] (01) at (-4.75, 2.6) {$\{0, 1\}$};
\node[nodestyle] (02) at (-2.85, 2.6) {$\{0, 2\}$};
\node[nodestyle,fill=orange] (03) at (-0.95, 2.6) {$\{0, 3\}$};
\node[nodestyle] (12) at (0.95, 2.6) {$\{1, 2\}$};
\node[nodestyle,fill=orange] (13) at (2.85, 2.6) {$\{1, 3\}$};
\node[nodestyle,fill=orange] (23) at (4.75, 2.6) {$\{2, 3\}$};
\node[nodestyle,fill=orange] (012) at (-2.85, 3.9) {$\{0, 1, 2\}$};
\node[nodestyle,fill=orange] (013) at (-0.95, 3.9) {$\{0, 1, 3\}$};
\node[nodestyle,fill=orange] (023) at (0.95, 3.9) {$\{0, 2, 3\}$};
\node[nodestyle,fill=orange] (123) at (2.85, 3.9) {$\{1, 2, 3\}$};
\node[nodestyle,fill=orange] (0123) at (0.0, 5.2) {$\{0, 1, 2, 3\}$};

 ===== DRAWING EDGES ====

\draw[black,thick] (emptyset) -- (0);
\draw[black,thick] (emptyset) -- (1);
\draw[black,thick] (emptyset) -- (2);
\draw[black,thick] (emptyset) -- (3);
\draw[black,thick,draw=black!50!green,dashed] (0) -- (01);
\draw[black,thick] (0) -- (02);
\draw[black,thick] (0) -- (03);
\draw[black,thick,draw=black!50!green,dashed] (1) -- (01);
\draw[black,thick,draw=black!50!green,dashed] (1) -- (12);
\draw[black,thick] (1) -- (13);
\draw[black,thick] (2) -- (02);
\draw[black,thick] (2) -- (12);
\draw[black,thick] (2) -- (23);
\draw[black,thick] (3) -- (03);
\draw[black,thick] (3) -- (13);
\draw[black,thick] (3) -- (23);
\draw[black,thick] (01) -- (012);
\draw[black,thick] (01) -- (013);
\draw[black,thick] (02) -- (012);
\draw[black,thick] (02) -- (023);
\draw[black,thick] (03) -- (013);
\draw[black,thick] (03) -- (023);
\draw[black,thick,draw=black!50!green,dashed] (12) -- (012);
\draw[black,thick] (12) -- (123);
\draw[black,thick] (13) -- (013);
\draw[black,thick] (13) -- (123);
\draw[black,thick] (23) -- (023);
\draw[black,thick] (23) -- (123);
\draw[black,thick] (012) -- (0123);
\draw[black,thick] (013) -- (0123);
\draw[black,thick] (023) -- (0123);
\draw[black,thick] (123) -- (0123);
	\end{tikzpicture}
	\end{scaletikzpicturetowidth}
	\caption{The colored graph $\mathbf{G}_V[\varphi_9]$.
	Note that it is different from $L_\cnf^{\varphi_9}$ (Figure~\ref{fig:lattice}).
	In particular, the semantics of the nodes is not the same.}
\label{fig:G_V-q9}
\end{figure}
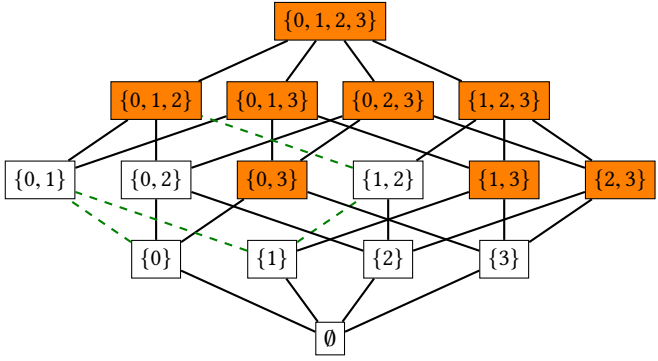

Then, we have $\varphi \rewr{$\pm$}{4}^* \varphi'$ whenever we can go from $\mathbf{G}_V[\varphi]$ to $\mathbf{G}_V[\varphi']$ 
by iteratively (1) coloring two uncolored adjacent nodes; or (2) uncoloring two adjacent colored nodes.
In particular, observe that this transformation never changes the Euler characteristic.
We illustrate the transformation in Figure~\ref{fig:chainswap} (ignore for now the last sentence in the description).

We can now show that if $\varphi \simeq \bot$, then~$\varphi$ is fragmentable:

\begin{proposition}
\label{prp:if-equiv-bot-then-fragmentable}
	If $\varphi \simeq \bot$ then~$\varphi$ is fragmentable.
\end{proposition}
\begin{proof}
	For some $n \in \mathbb{N}$, we have $\bot = \varphi_0 \rewr{$\pm(\nu_1,l_1)$}{4} \cdots \rewr{$\pm(\nu_n,l_n)$}{4} \varphi_n = \varphi$ for some valuations and variables $\nu_i,l_i$.
	We show by induction on $0 \leq i \leq n$ that~$\varphi_i$ is fragmentable by constructing a~$\lnot$-$\lor$-template~$T_i$ having~$i+1$ holes~$l_0,\ldots,l_i$ and degenerate
	Boolean functions~$\psi_j$ for~$0 \leq j \leq i$ such that~$T_i[\psi_0,\ldots,\psi_i]$ is deterministic and represents~$\varphi_i$.
	The base case of $i=0$ is trivial since~$\bot$ is degenerate, hence fragmentable.
	For the inductive case, supose $i > 0$ with $\varphi_{i-1}$ fragmentable, and let~$T_{i-1}$ and~$\psi_0,\ldots,\psi_{i-1}$ be a template and degenerate functions
	witnessing that~$\varphi_{i-1}$ is fragmentable.
	Define~$\psi_i$ by $\sat(\psi_i) \defeq \{\nu_i,\nu_i^{(l_i)}\}$.
	Surely~$\psi_i$ is degenerate (it does not depend on variable~$l_i$).
	We then distinguish the two possible cases:
	\begin{itemize}
		\item we have $\varphi_{i-1} \rewr{$+(\nu_i,l_i)$}{4} \varphi_i$.
			Then we can write~$\varphi_i$ as the deterministic disjunction $\varphi_{i-1} \lor \psi_i $.
			But then we can define~$T_i$ to be the template~$T_i \defeq T_{i-1} \lor l_i$, 
			and one can easily check that~$T_i[\psi_0,\ldots,\psi_i]$ is deterministic and represents~$\varphi_i$.
			Therefore,~$\varphi_i$ is indeed fragmentable.
		\item we have $\varphi_{i-1} \rewr{$-(\nu_i,l_i)$}{4} \varphi_i$.
			Then we can write $\varphi_i$ as $\lnot(\lnot\varphi_{i-1} \lor \psi_i)$, with the disjunction being deterministic.
			But then we can define~$T_i$ to be the template~$T_i \defeq \lnot(\lnot T_{i-1} \lor l_i)$, and once again
			it is direct that~$T_i[\psi_0,\ldots,\psi_i]$ is deterministic and represents~$\varphi_i$.
			Hence~$\varphi_i$ is fragmentable.\qedhere
	\end{itemize}
\end{proof}

\subsection{If $\eul(\varphi)=0$ then $\varphi \simeq \bot$}
\label{sec:if-euler-null-then-equiv-bot}

The missing piece to show Proposition~\ref{prp:if-euler-null-then-fragmentable} is to prove that whenever $\eul(\varphi)=0$ then~$\varphi$ is
equivalent to~$\bot$.
This is what we do in this section.
Formally:

\begin{proposition}
\label{prp:if-euler-null-then-equiv-bot}
	If $\eul(\varphi)=0$ then $\varphi \simeq \bot$.
\end{proposition}

In order to do that, we need two simple lemmas, which we will also reuse in the next section.
The first is what we call the \emph{chaining lemma}:

\begin{lemma}[(Chaining lemma)]
\label{lem:chaining}
	Let~$\varphi$ be a Boolean function, and $\nu \neq \nu'$ be two valuations such that there is a simple path
	$\nu = \nu_0 - \cdots - \nu_{n+1} = \nu'$ from $\nu$ to $\nu'$ in 
	$\mathbf{G}_V$ with $n \geq 0$ and $\nu_i \notin \sat(\varphi)$ for $1 \leq i \leq n$.
	Then we have the following:
	\begin{description}
		\item[\bf Chainkilling.] If $(-1)^{|\nu|} \neq (-1)^{|\nu'|}$ (i.e., $n$ is even) and $\{\nu,\nu'\} \subseteq \sat(\varphi)$ then,
			defining~$\varphi'$ by $\sat(\varphi') \defeq \sat(\varphi) \setminus \{\nu,\nu'\}$, we have $\varphi \rewr{$\pm$}{4}^* \varphi'$.
			In other words, we can uncolor both $\nu$ and $\nu'$.
			We say that \emph{we have chainkilled~$\nu$ and~$\nu'$}.
		\item[\bf Chainswapping.] If $(-1)^{|\nu|} = (-1)^{|\nu'|}$ (i.e., $n$ is odd) and $\nu \in \sat(\varphi)$ and $\nu' \notin \sat(\varphi)$ then,
			defining~$\varphi'$ by $\sat(\varphi') \defeq (\sat(\varphi) \setminus \{\nu\}) \cup \{\nu'\}$, we have $\varphi \rewr{$\pm$}{4}^* \varphi'$.
			In other words, we can uncolor $\nu$ and color~$\nu'$.
			We say that \emph{we have chainswapped~$\nu$ to~$\nu'$}.
	\end{description}
\end{lemma}
\begin{proof}
	We only explain the chainkilling part, as chainswapping works similarly.
	Let $n=2i$. For $0 \leq j < i$, do the following: color the nodes $\nu_{2j+1}$ and $\nu_{2j+2}$ and then uncolor the nodes
	$\nu_{2j}$ and $\nu_{2j+1}$. Finally, uncolor the nodes~$\nu_{2i}$ and~$\nu_{2i+1}$.
	(Or, alternatively, first color all the nodes on the path and then uncolor everything.)
\end{proof}

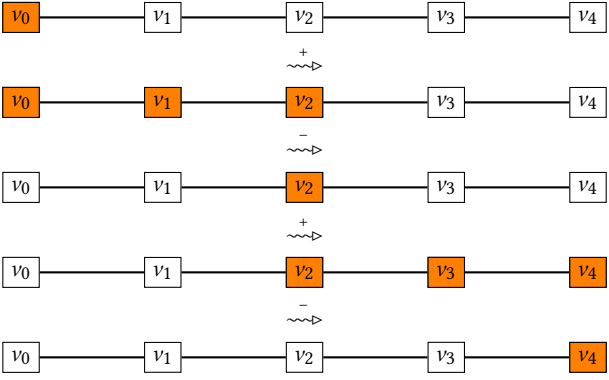
\begin{figure}
\centering
\begin{scaletikzpicturetowidth}{\linewidth}
\begin{tikzpicture}[scale=\tikzscale]	
\tikzset{nodestyle/.style={draw,rectangle}}
\foreach \i [evaluate=\i as \ii using \i*0.6] in {0,1,2,3,4} {
\foreach \j in {0,1,2,3,4} {
	\node[name=n\i\j,nodestyle] at (\j, \ii) {$\nu_\j$};
		}
	\foreach \k/\l in {0/1,1/2,2/3,3/4} {
			\draw[black, thick] (n\i\k) -- (n\i\l);
			}
			}

\foreach \i [evaluate=\i as \ii using \i*0.6] in {0,1,2,3} {
	\ifodd\i
	\node at (2,\ii+0.3) {$\rewr{$+$}{3}$};
	\else
	\node at (2,\ii+0.3) {$\rewr{$-$}{3}$};
	\fi
	}

	\node[nodestyle,fill=orange] at (0, 4*0.6) {$\nu_0$};

	\node[nodestyle,fill=orange] at (0, 3*0.6) {$\nu_0$};
	\node[nodestyle,fill=orange] at (1, 3*0.6) {$\nu_1$};
	\node[nodestyle,fill=orange] at (2, 3*0.6) {$\nu_2$};

	\node[nodestyle,fill=orange] at (2, 2*0.6) {$\nu_2$};

	\node[nodestyle,fill=orange] at (2, 1*0.6) {$\nu_2$};
	\node[nodestyle,fill=orange] at (3, 1*0.6) {$\nu_3$};
	\node[nodestyle,fill=orange] at (4, 1*0.6) {$\nu_4$};

	\node[nodestyle,fill=orange] at (4, 0*0.6) {$\nu_4$};
\end{tikzpicture}
\end{scaletikzpicturetowidth}
	\caption{Imagine that the path at the top is a subgraph of the colored graph $\mathbf{G}_V[\varphi]$ for some function~$\varphi$ (for instance, the dashed green path
	in Figure~\ref{fig:G_V-q9}).
	The consecutive subgraphs are obtained by single steps of our transformation.
	The total transformation also illustrates what we call a chainswap (Definition~\ref{lem:chaining}).}
\label{fig:chainswap}
\end{figure}

Figure~\ref{fig:chainswap} illustrates a chainswap.
In this section we will only need the chainkilling part of the lemma, but we will use chainswapping in the next section.
The second lemma that we need is the following (in this section we only need its instantiation when $\eul(\varphi)=0$), which we call the \emph{fetching lemma}:

\begin{lemma}[(Fetching lemma)]
\label{lem:fetch}
	Let~$\varphi$ be such that $\#\varphi \neq |\eul(\varphi)|$.
	Then there exist satisfying valuations $\nu,\nu'$ of~$\varphi$ with $(-1)^{|\nu|} \neq (-1)^{|\nu'|}$ and a simple path
	$\nu = \nu_0 - \cdots - \nu_{n+1} = \nu'$ from $\nu$ to $\nu'$ in 
	$\mathbf{G}_V$ (hence with $n$ even) such that $\nu_i \notin \sat(\varphi)$ for $1 \leq i \leq n$.
\end{lemma}
\begin{proof}
	Since $\#\varphi \neq |\eul(\varphi)|$, there exist satisfying valuations $\nu'',\nu'''$ of~$\varphi$ with $(-1)^{|\nu''|} \neq (-1)^{|\nu'''|}$.
	Let $\nu'' = \nu''_0 - \cdots - \nu''_{m+1} = \nu''''$ be an arbitrary simple path from $\nu''$ to $\nu'''$ in 
	$\mathbf{G}_V$ (such a path clearly exists because $\mathbf{G}_V$ is a connected graph).
	Now, let $i \defeq \max(0 \leq j \leq m \mid (-1)^{|\nu''_j|} = (-1)^{|\nu''|} \text{ and }\nu''_j \models \varphi)$, and then
	let $i' \defeq \min(i < j \leq m+1 \mid (-1)^{|\nu''_{j}|} = (-1)^{|\nu'''|} \text{ and }\nu''_j \models \varphi)$.
	Then we can take $\nu$ to be $\nu''_i$ and $\nu'$ to be $\nu''_{i'}$, which satisfy the desired property.
\end{proof}

Given the right circumstances ($\#\varphi \neq |\eul(\varphi)|$), this lemma fetches two valuations $\nu,\nu'$ for us to chainkill them.
With these tools in place, it is now easy to prove Proposition~\ref{prp:if-euler-null-then-equiv-bot}:

\begin{proof}[Proof of Proposition~\ref{prp:if-euler-null-then-equiv-bot}]
	Let~$\varphi$ be such that $\eul(\varphi)=0$.
	We prove the claim by induction on $\#\varphi$. When $\#\varphi=0$ then we have $\varphi = \bot$, hence a fortiori $\varphi \simeq \bot$.
	Assume now that $\#\varphi > 0$.
	Lemma~\ref{lem:fetch} fetches two satisfying valuations $\nu,\nu$ of~$\varphi$, which we then chainkill.
	Let~$\varphi'$ be the function obtained. We again have $\eul(\varphi')=0$, because the transformation $\rewr{$\pm$}{4}$ does not change the Euler characteristic.
	Moreover, we have $\#\varphi' = \#\varphi -2$, so that $\varphi' \simeq \bot$ by induction hypothesis, and therefore $\varphi \simeq \bot$ as well.
\end{proof}

Notice how the same valuation could be colored
and uncolored multiple times, in case it appears on the paths of different pairs of valuations that are chainkilled.
Also, and although this is of absolutely no use to us, we note here the amusing fact that an analogue of Proposition~\ref{prp:if-euler-null-then-equiv-bot} works for any connected bipartite colored graph.
Moreover, by inspection of the proofs, we also obtain the following:

\begin{corollary}
\label{cor:computable}
	It is computable, given a fragmentable Boolean function~$\varphi$, to find a $\lnot$-$\lor$-template $T$
	and degenerate Boolean functions $\varphi_0,\ldots,\varphi_n$ such that $T[\varphi_0,\ldots,\varphi_n]$
	if deterministic and equivalent to~$\varphi$.
\end{corollary}

This concludes the upper-bound results of this paper.
In the next section, we reuse the tools that we have developed so far to carry a detailed analysis of the remaining cases ($\eul(\varphi) \neq 0$).

\section{Equivalences and Hardness}
  \label{sec:hardness}
  The results of the last section imply that, for any two Boolean functions $\varphi,\varphi'$ with $\eul(\varphi)=\eul(\varphi')=0$, we have $\varphi \simeq \varphi'$.
Indeed, by Proposition~\ref{prp:if-euler-null-then-equiv-bot} we have $\varphi \simeq \bot$ and $\varphi' \simeq \bot$, implying $\varphi \simeq \varphi'$.
In fact, we can show a more general statement about our transformation. We claim:

\begin{proposition}
\label{prp:same-euler-iff-simeq}
	We have $\eul(\varphi)=\eul(\varphi')$ if and only if $\varphi \simeq \varphi'$.
\end{proposition}

The proof is similar to that of Proposition~\ref{prp:if-euler-null-then-equiv-bot}, but is more involved.
We first discuss the consequences of Proposition~\ref{prp:same-euler-iff-simeq} in terms of tractability and compilation of the $\mathcal{H}$-queries
in Section~\ref{sec:equiv-consequences}, and then prove the proposition in Section~\ref{sec:proof-equiv}.

\subsection{Consequences}
\label{sec:equiv-consequences}

Proposition~\ref{prp:same-euler-iff-simeq} allows us to prove:

\begin{theorem}
\label{thm:equivalences}
	Let $\varphi,\varphi'$ such that $\eul(\varphi)=\eul(\varphi')$. Then the following hold:
	\begin{description}
		\item[(a)] we have $\pqe(Q_\varphi) \equiv_\mathrm{T} \pqe(Q_{\varphi'})$;
		\item[(b)] we have $Q_\varphi \in \ddptime$ iff $Q_{\varphi'} \in \ddptime$;
		\item[(c)] we have $Q_\varphi \in \ddpsize$ iff $Q_{\varphi'} \in \ddpsize$.
	\end{description}
\end{theorem}
\begin{proof}
	The proof is similar to that of Proposition~\ref{prp:if-equiv-bot-then-fragmentable}.
	We sketch the only-if directions of (b) and (c), but (a) works similarly.
	By Proposition~\ref{prp:same-euler-iff-simeq}, we have $\varphi = \varphi_0 \rewr{$\pm(\nu_1,l_1)$}{4} \cdots \rewr{$\pm(\nu_n,l_n)$}{4} \varphi_n = \varphi'$
	for some $n \in \mathbb{N}$ and valuations and variables $\nu_i,l_i$.
	We show by induction on $0 \leq i \leq n$ that $Q_{\varphi_i} \in \ddptime$.
	The case $i=0$ holds by assumption.
	For the inductive case we proceed just as in Proposition~\ref{prp:if-equiv-bot-then-fragmentable}, but this time we use the 
	degeneracy of~$\psi_i$ to obtain $Q_{\psi_i} \in \ddptime$
	by Proposition~\ref{prp:if-degenerate-then-obddptime}.
	The whole construction can clearly be carried out in PTIME data complexity.
\end{proof}

Observe how each item says something different: as far as we can tell, and
except when $\eul(\varphi)=\eul(\varphi')=0$, in which case all three items are equivalent
since $\pqe(Q_\varphi) \in \ddptime$ by Section~\ref{sec:upper_bound},
there is no obvious implication between any two of them.

Also observe that Theorem~\ref{thm:equivalences}, together with the observations that $\eul(\bot)=0$ and $Q_\bot \in \ddptime$, imply
our main result of Theorem~\ref{thm:if-euler-null-then-ddptime}.
We decided to present them separately to make the proofs more modular, so that someone only interested into the upper bound results of the paper can
understand the relevant parts more easily (since, as we said, Proposition~\ref{prp:same-euler-iff-simeq} requires more work; see also Figure~\ref{fig:proof-structure}).

We must now discuss a result of~\cite{beame2017exact} concerning FBDDs for the $\mathcal{H}$-queries which, on the surface, seems closely related to the last two items of Theorem~\ref{thm:equivalences}.
There, the authors show the following (we paraphrase to fit our notation):

\begin{theorem}[{\cite[Theorem~3.9]{beame2017exact}}]
\label{thm:beame-fbdds}
	Let~$\varphi$ be nondegenerate, and $\varphi'$ be an arbitrary Boolean function.
	Then any FBDD $F$ representing the lineage of $Q_\varphi$ on a database $D$
	can be transformed into an FBDD $F'$ representing the lineage of $Q_{\varphi'}$ on the same database,
	with only a polynomial increase in size.
\end{theorem}

First, we note that the proof techniques are seemingly very different: while we did not need to work at the tuple-level, the proof of Theorem~3.9 proceeds by
chirugically examining the FBDD $F$ in order to transform it into $F'$.
In fact, one could even argue that the two results are inherently incomparable, for the following reasons:
\begin{enumerate}
	\item Theorem~\ref{thm:equivalences}, items (b) and (c), do not hold when applied to FBDDs.
		Indeed, take $\varphi \defeq \bot$ and $\varphi'$ to be $\varphi_9$ from	Example~\ref{expl:q9}.
		We have $\eul(\bot)=\eul(\varphi_9) = 0$.
		Yet, $Q_\bot \in \fbddptime$ (clearly), but $Q_{\varphi_9} \notin \fbddpsize$ by the lower bounds of~\cite{beame2017exact}. 
	\item Less convincingly, if we could show a version of Theorem~\ref{thm:beame-fbdds} with \mbox{d-Ds} instead of FBDDs, then \emph{all}
		$\mathcal{H}$-queries (for $k \geq 2$) would be in $\ddpsize$, including those that are \#P-hard by~\cite{dalvi2012dichotomy}.
		Indeed, for all $k \in \mathbb{N}_{\geq 2}$, there clearly exists a query $Q_\varphi \in \mathcal{H}_k$ such that
		$\varphi$ is nondegenerate and $\eul(\varphi)=0$.
		Then by Theorem~\ref{thm:if-euler-null-then-ddptime} we have $Q_\varphi \in \ddptime$, hence any other $\mathcal{H}_k$-query would
		also be in $\ddpsize$.
		This would not be in contradiction with anything, as far as we know, but it would still be quite surprising.
\end{enumerate}

Second, by combining Theorem~\ref{thm:beame-fbdds} with an exponential lower bound on FBDD representations of one particular nondegenerate
$\mathcal{H}$-query $Q_{\varphi_{\mathrm{big-FBDDs}}}$, the authors of~\cite{beame2017exact} obtain an exponential lower bound on
FBDDs for all nondegenerate $\mathcal{H}$-queries.
In our case we observe a similar phenomenon with Theorem~\ref{thm:equivalences}, where an exponential lower bound on d-D
representations for a query $Q_{\varphi_{\mathrm{big-d-Ds}}}$ would
also apply to all queries with same Euler characteristic.
However, the existence of such a query $Q_{\varphi_{\mathrm{big-d-Ds}}}$ would answer a long-standing open problem in knowledge compilation:
that of separating \mbox{d-Ds} (or even d-DNNFs) with
DNFs~\cite{darwiche2002knowledge,bova2016knowledge,amarilli2019connecting}.
This is because the lineage of any UCQ on a database can always be computed in PTIME as a~DNF.

Nevertheless, we can still use Theorem~\ref{thm:equivalences} to show a hardness result extending that of ~\cite{dalvi2012dichotomy} for $\mathcal{H}^+$-queries (though this was not the primary goal of this paper). Specifically, we show:

\begin{proposition}
\label{prp:hardness}
	Let~$\varphi$ be a Boolean function (not necessarily monotone) with~$\eul(\varphi) \neq 0$
	and such that~$\min(\eul(\varphi) : \varphi \text{ is monotone}) \leq \eul(\varphi) \leq \max(\eul(\varphi) : \varphi \text{ is monotone})$.
	Then $\pqe(Q_\varphi)$ is \#P-hard.
\end{proposition}
\begin{proof}
	Let~$\varphi$ be such a Boolean function. We show in Appendix~\ref{apx:max-euler} that there exists a monotone Boolean function $\varphi_{\mathrm{mon}}$
	such that $\eul(\varphi_{\mathrm{mon}}) = \eul(\varphi)$.
	But then we have $\eul(\varphi_{\mathrm{mon}}) \neq 0$, hence by Corollary~\ref{cor:rephrase} we have that $\pqe(Q_{\varphi_{\mathrm{mon}}})$ is \#P-hard,
	and by Theorem~\ref{thm:equivalences}, item (a), it holds that $\pqe(Q_\varphi)$ is \#P-hard as well.
\end{proof}

\begin{toappendix}
	\label{apx:max-euler}
	In this section we prove the following, which we used in the proof of Proposition~\ref{prp:hardness} (remember, once again, that~$k \in \mathbb{N}_{\geq 1}$
	and the set of variables~$V = \{0,\ldots,k\}$ are fixed):
	
	\begin{lemma}
	\label{lem:max-euler}
		Let~$\varphi$ be a Boolean function (not necessarily monotone) with~$\eul(\varphi) \neq 0$
		and such that~$\min(\eul(\varphi) : \varphi \text{ is monotone}) \leq \eul(\varphi) \leq \max(\eul(\varphi) : \varphi \text{ is monotone})$.
		Then there exists a monotone Boolean function $\varphi_{\mathrm{mon}}$
		such that $\eul(\varphi_{\mathrm{mon}}) = \eul(\varphi)$.
	\end{lemma}
	\begin{proof}
		We will assume wlog that~$\eul(\varphi) > 0$, as the case of~$\eul(\varphi) < 0$ is symmetric.
		The idea of the proof is to start from a monotone Boolean function~$\varphi_M$ that maximizes~$\eul(\varphi_M)$, and then
		to modify~$\varphi_M$ (specifically, we will remove satisfying valuations) to show that
		all the intermediate values between~$0$ and~$\eul(\varphi_M)$ for the Euler characteristic
		(hence, also the value~$\eul(\varphi)$) are achievable by a monotone Boolean function.
		More formally, consider the Boolean functions~$\psi_0,\ldots,\psi_{|\sat(\varphi_M)|}$, obtained
		by starting with~$\psi_0 \defeq \varphi_M$ and iteratively removing exactly one satisfying assignment of maximal size.
		Then each~$\psi_i$ for~$ 0 \leq i \leq |\sat(\varphi_M)|$ is a monotone Boolean function, and we have that~$\psi_{|\sat(\varphi_M)|} = \bot$.
		Moreover, the Euler characteristic of~$\psi_i,\psi_{i+1}$ differ only by one.
		Since we have~$\eul(\bot)=0$, and since we know that $0 < \eul(\varphi) \leq \eul(\psi_0)$,
		there exists a~$\psi_i$	such that~$\eul(\psi_i) = \eul(\varphi)$, and we can then take~$\psi_i$ to be~$\varphi_{\mathrm{mon}}$.
		This concludes the proof.
	\end{proof}
	
	For completeness (but this is of no use to us), we note here that the monotone Boolean functions~$\varphi_M$ such that~$|\eul(\varphi_M)|$ is maximized
	are characterized precisely in~\cite{bjorner1988extended}.
	Since~\cite{bjorner1988extended} is all about simplicial complexes, we will paraphrase their result in terms of Boolean functions here.
	To do it correctly,
	we must keep in mind that (1) for our purposes, a “simplicial complex” is the same thing as “the negation of a monotone Boolean function”, so we need to reverse the powerset; and (2) in simplicial complexes terminology, the \emph{dimension} of a face~$\nu \subseteq V$ is~$|\nu| -1$. We then have:

	\begin{theorem}[({See \cite[Theorem 1.4]{bjorner1988extended}})]
	\label{thm:max-euler}
		A monotone Boolean function~$\varphi_M$ such that~$|\eul(\varphi_M)| = \mathrm{M}(k)$ can only be obtained as follows:
		\begin{itemize}
			\item If~$k$ is even, then take $\sat(\varphi_M) \defeq \{\nu \subseteq V \mid |\nu| \geq n/2 + 1\}$;
			\item If~$k$ is odd, then take $\sat(\varphi_M) \defeq \{\nu \subseteq V \mid |\nu| \geq (n-1)/2 +1\}$, or take
				$\sat(\varphi_M) \defeq \{\nu \subseteq V \mid |\nu| \geq (n+1)/2 +2\}$.
		\end{itemize}
	\end{theorem}
	\clearpage
\end{toappendix}

We represented by the dashed red rectangle in Figure~\ref{fig:full-picture} all the $\mathcal{H}$-queries to which this result applies.
This unfortunately does not apply to all the $\mathcal{H}$-queries $Q_\varphi$ with $\eul(\varphi) \neq 0$:
for instance, consider the function $\varphi_{\mathrm{max-Euler}}$ whose satisfying
valuations are exactly all the valuations of even size.
Then we have $\eul(\varphi_{\mathrm{max-Euler}}) = 2^k$, and this value is not reachable by a monotone function.

\subsection{Proof of Proposition~\ref{prp:same-euler-iff-simeq}}
\label{sec:proof-equiv}

As we have already observed, the transformation does not change the Euler characteristic, so we only need to show
the “only if” part of Proposition~\ref{prp:same-euler-iff-simeq}, i.e., that if $\eul(\varphi)=\eul(\varphi')$ then $\varphi \simeq \varphi$.
Furthermore, since $\eul(\varphi)=-\eul(\lnot\varphi)$, it is enough to show it when $\eul(\varphi) \geq 0$.
We proceed in three steps.
The first step is to transform the functions so that they have only satisfying valuations of even size.
Formally:

\begin{lemma}
\label{lem:simeq-min}
	Let~$\varphi$ a Boolean function with $\eul(\varphi) \geq 0$.
	Then there exists $\varphi_{\mathrm{min}} \simeq \varphi$ having only
	satisfying valuations of even size.
\end{lemma}
\begin{proof}
	We proceed as in the proof of Proposition~\ref{prp:if-euler-null-then-equiv-bot}: until~$\varphi$ has satisfying valuations
	of odd size, we use Lemma~\ref{lem:fetch} to
	fetch two valuations $\nu,\nu'$, and then we chainkill them with Lemma~\ref{lem:chaining}.
	Indeed, this always decreases the number of valuations of odd size by one.
\end{proof}

By applying Lemma~\ref{lem:simeq-min} to~$\varphi$ and $\varphi'$, we are left with two functions $\varphi_{\mathrm{min}}$ and $\varphi_{\mathrm{min}}'$, both
having only valuations of even size, and with $\#\varphi_{\mathrm{min}} = \#\varphi_{\mathrm{min}}'$.
The second step is then to transform these into what we call a \emph{canonical form}.

\begin{definition}
\label{def:canonical}
	We say that a Boolean function~$\varphi$ is in canonical form when (1)~$\varphi$ only has satisfying valuations of even size;
	and (2) there does not exist two valuations $\nu,\nu'$ of even size with $|\nu'| < |\nu|$ and $\nu \in \sat(\varphi)$ but $\nu' \notin \sat(\varphi)$
	(in other words, the satisfying valuations of~$\varphi$ are the smallest possible). Let us call such a pair of valuations a \emph{bad pair}.
\end{definition}

We then show:

\begin{lemma}
\label{lem:simeq-check}
	Let~$\varphi$ be a Boolean function having only satisfying valuations of even size.
	Then there exists a Boolean function $\check{\varphi}$ in canonical form such that $\check{\varphi} \simeq \varphi$.
\end{lemma}
\begin{proof}
	We prove the claim by induction on the number $B_\varphi$ of bad pairs.
	If $B_\varphi = 0$ then we are done, so assume $B_\varphi >0$, and let~$(\nu,\nu')$ be a bad pair of~$\varphi$.
	We distinguish two cases:
	\begin{itemize}
		\item we have $\nu' \subseteq \nu$. Then there exists a descending path $p$ from~$\nu$ to~$\nu'$ in $\mathbf{G}_V$.
			By “descending” we mean that
			the sizes of the valuations along the path are strictly decreasing (by $1$, of course).
			Let~$\nu_i$ be a valuation on that path that satisfies~$\varphi$ and such that no intermediary valuation on $p$ from~$\nu_i$ to~$\nu'$
			satisfies~$\varphi$.
			Such a~$\nu_i$ clearly exists, since in the worst case we can take $\nu_i \defeq \nu$.
			Observe that $(\nu_i,\nu')$ is a bad pair.
			We now chainswap~$\nu_i$ to~$\nu'$. Let~$\varphi'$ be the function obtained.
			It is clear that~$\varphi'$ only has satisfying valuations of even size, and that $B_{\varphi'} = B_{\varphi} -1$
			(since we have eliminated
			the bad pair $(\nu_i,\nu')$ without introducing others), so that the claim holds by induction hypothesis.
		\item we have $\nu' \not\subseteq \nu$.
			Then consider another valuation~$\nu''$ 
			with $|\nu'|=|\nu''|$ and such that there is a descending path from~$\nu$ to~$\nu''$ in $\mathbf{G}_V$
			(i.e., such that $\nu'' \subseteq \nu$: there clearly exist one such~$\nu''$).
			If~$\nu''$ does not satisfy~$\varphi$, simply use the preceding item with~$\nu''$ instead of~$\nu'$, and we are done.
			If~$\nu''$ satisfies~$\varphi$, then consider a path
			$\nu'' = \nu^=_0 - \nu^+_1 - \nu^=_1 - \cdots - \nu^+_j - \nu^=_j - \cdots \nu^+_n - \nu^=_n = \nu'$
			from~$\nu''$ to~$\nu'$ that alternates between valuations~$\nu^=_j$ of size $|\nu''|=|\nu'|$ and valuations~$\nu^+_j$ of size $|\nu''|+1$. 
			Again, such a path clearly exists.
			For instance if $\nu = \{0,1,2,3\}$ and $\nu' = \{4,5\}$ and $\nu'' = \{2,3\}$, we have the path
			$\nu'' = \{2,3\} - \{2,3,4\} - \{3,4\} - \{3,4,5\} - \{4,5\} = \nu'$.
			Now, consider the smallest $j$ such that $\nu^=_j \notin \sat(\varphi)$ (which is defined since $\nu' = \nu^=_n \notin \sat(\varphi)$).
			We now chainswap~$\nu''$ to~$\nu^=_j$, which does not introduce any bad pair, and preserves the
			fact that all satisfying valuations are of even size.
			Then, we can simply use the preceding item with~$\nu''$ instead of~$\nu'$.\qedhere
			\end{itemize}
	\end{proof}

	We now apply Lemma~\ref{lem:simeq-check} to $\varphi_{\mathrm{min}}$ and $\varphi_{\mathrm{min}}'$, thus obtaining~$\check{\varphi}_{\mathrm{min}}$
	and $\check{\varphi}_{\mathrm{min}}'$, both in canonical form, and with
	$\#\check{\varphi}_{\mathrm{min}} = \#\check{\varphi}_{\mathrm{min}}'$ again.
	The third step is to show that $\check{\varphi}_{\mathrm{min}} \simeq \check{\varphi}_{\mathrm{min}}'$.
	Letting $M \defeq \max(|\nu| : \nu \models \check{\varphi}_{\mathrm{min}})$ and 
	$M' \defeq \max(|\nu| : \nu \models \check{\varphi}_{\mathrm{min}}')$, we necessarily have~$M=M'$ because the functions are in canonical form and
	have the same number of satisfying valuations.
	If~$M = |V|$ then clearly $\check{\varphi}_{\mathrm{min}} = \check{\varphi}_{\mathrm{min}}'$, because then
	$\sat(\check{\varphi}_{\mathrm{min}})$ and $\sat(\check{\varphi}_{\mathrm{min}}')$ both contain exactly \emph{all} the valuations of even size.
	Otherwise, we can use valuations of size~$M+1$ to transform, say, $\check{\varphi}_{\mathrm{min}}$ into $\check{\varphi}_{\mathrm{min}}'$,
	by iteratively doing the following:
	\begin{enumerate}
		\item pick a valuation~$\nu \in \sat(\check{\varphi}_{\mathrm{min}}) \setminus \sat(\check{\varphi}_{\mathrm{min}}')$ and
			a valuation~$\nu' \in \sat(\check{\varphi}_{\mathrm{min}}') \setminus \sat(\check{\varphi}_{\mathrm{min}})$;
			these exist and are necessarily of size~$M$;
		\item let $\nu = \nu^=_0 - \nu^+_1 - \nu^=_1 - \cdots - \nu^+_j - \nu^=_j - \cdots \nu^+_n - \nu^=_n = \nu'$ be a simple path
			from~$\nu$ to~$\nu'$ that alternates between valuations~$\nu^=_j$ of size $|\nu|=|\nu'|=M$ and valuations~$\nu^+_j$ of size~$M+1$. 
			Let $0 = j_0 < \ldots <j_m < n$ be all the indices such that $\nu^=_{j_l} \models \check{\varphi}_{\mathrm{min}}$.
			We then chainswap~$\nu^=_{j_m}$ to~$\nu'$, and then for $0 \leq p < m$ we chainswap~$\nu^=_{j_p}$ to~$\nu^=_{j_{p+1}}$.
			The total transformation then simply amounts to uncoloring~$\nu$ and coloring~$\nu'$.
	\end{enumerate}
	Therefore, we indeed have $\check{\varphi}_{\mathrm{min}} \simeq \check{\varphi}_{\mathrm{min}}'$,
	which implies $\varphi_{\mathrm{min}} \simeq \varphi_{\mathrm{min}}'$ and then $\varphi \simeq \varphi'$,
	concluding the proof of Proposition~\ref{prp:same-euler-iff-simeq}.

\section{Open Problems and Future Work}
  \label{sec:open}
  In this section we mention some of the numerous questions that this work must leave behind.
Some of them are directly related to other open problems in knowledge compilation, while others seem specific to probabilistic databases.
We also justify the definition of our transformation.

\paragraph{\bf Hardness and lower bounds}
We start with the question of showing hardness for the queries in the dotted gray rectangle of Figure~\ref{fig:full-picture}:

\begin{openpb}
	\label{open:if-euler-not-null-then-hard}
	Show \#P-hardness of all $\mathcal{H}$-queries $Q_\varphi$ with $\eul(\varphi) \neq 0$.
\end{openpb}

This would require adapting the proofs of~\cite{dalvi2012dichotomy} to UCQs with negations, which appears to be challenging.
Note that, thanks to Theorem~\ref{thm:equivalences} and Lemmas~\ref{lem:simeq-min} and~\ref{lem:simeq-check}, it
would be enough to show the hardness of the queries in canonical form.
In an orthogonal direction, and as we have mentioned already, Theorems~\ref{thm:if-euler-null-then-ddptime} and \ref{thm:equivalences} are begging for a counterpart lower bound:

\begin{openpb}
	\label{open:lower-dd}
	Show superpolynomial lower bounds for d-D representations of the lineages of all $\mathcal{H}$-queries $Q_\varphi$ with $\eul(\varphi) \neq 0$.
\end{openpb}

However, a lower bound on \mbox{d-Ds} seems far out of reach with current techniques.
Instead, a framework to show lower bounds on d-DNNFs has recently been introduced in~\cite{bova2016knowledge}, using
tools from communication complexity.
Hence, maybe showing a lower bound on d-DNNFs representations could be an easier target
(though this would still answer an important open problem in knowledge compilation).

\paragraph{\bf Using fewer negations}
A brief inspection of our proofs indicates that we have spent a generous number of negations to construct the \mbox{d-Ds}.
One can readily wonder if that was really necessary.
In this regard, an easy observation is that if $\varphi \rewr{$-$}{4}^* \bot$, then $Q_\varphi \in \ddnnfptime$,
and that if $\varphi \rewr{$+$}{4}^* \top$ then $\lnot Q_\varphi \in \ddnnfptime$.
This was actually the approach taken in~\cite{monet2018towards}.
The facts $\varphi \rewr{$-$}{4}^* \bot$ and $\varphi \rewr{$+$}{4}^* \top$ can be reformulated using more standard notions of graph theory: we have $\varphi \rewr{$-$}{4}^* \bot$
iff the subgraph of $\mathbf{G}_V[\varphi]$ induced by the colored nodes has a perfect matching, and $\varphi \rewr{$+$}{4}^* \top$
iff that induced by the non-colored nodes has a perfect matching. We then conjecture the following:

\begin{conjecture}[(See also~\cite{cstheory_perfect_matchings,monet2018towards})]
	\label{open:perfect-matchings}
	If~$\varphi$ is monotone and $\eul(\varphi)=0$, then the subgraph of $\mathbf{G}_V[\varphi]$ induced by the colored nodes,
	\emph{or} that induced by the non-colored nodes,
	has a perfect matching.
\end{conjecture}

First, we note that the claim does not hold if we do not impose~$\varphi$ to be monotone.
Indeed, consider the function $\varphi_{\mathrm{no-PM}}$ for $k=4$, whose graph $\mathbf{G}_V[\varphi_{\mathrm{no-PM}}]$ we have depicted in Figure~\ref{fig:non-monotone-no-PM}.
\begin{figure*}[h]
\centering
	\begin{scaletikzpicturetowidth}{\linewidth}
	\begin{tikzpicture}[scale=\tikzscale]
	\tikzset{nodestyle/.style={draw,rectangle}}
	% Generated with the python script in ../code/generate_tikz_for_posets.py, using the file ../code/functions/phi_matchings-poset

 ===== DRAWING NODES ====

\node[nodestyle] (emptyset) at (0.0, 0.0) {$\emptyset$};
\node[nodestyle] (0) at (-3.8, 1.5) {$0$};
\node[nodestyle] (1) at (-1.9, 1.5) {$1$};
\node[nodestyle] (2) at (0.0, 1.5) {$2$};
\node[nodestyle] (3) at (1.9, 1.5) {$3$};
\node[nodestyle] (4) at (3.8, 1.5) {$4$};
\node[nodestyle] (01) at (-8.55, 3.0) {$01$};
\node[nodestyle] (02) at (-6.65, 3.0) {$02$};
\node[nodestyle,fill=orange] (03) at (-4.75, 3.0) {$03$};
\node[nodestyle,fill=orange] (04) at (-2.85, 3.0) {$04$};
\node[nodestyle] (12) at (-0.95, 3.0) {$12$};
\node[nodestyle] (13) at (0.95, 3.0) {$13$};
\node[nodestyle] (14) at (2.85, 3.0) {$14$};
\node[nodestyle] (23) at (4.75, 3.0) {$23$};
\node[nodestyle] (24) at (6.65, 3.0) {$24$};
\node[nodestyle,fill=orange] (34) at (8.55, 3.0) {$34$};
\node[nodestyle,fill=orange] (012) at (-8.55, 4.5) {$012$};
\node[nodestyle,fill=orange] (013) at (-6.65, 4.5) {$013$};
\node[nodestyle,fill=orange] (014) at (-4.75, 4.5) {$014$};
\node[nodestyle] (023) at (-2.85, 4.5) {$023$};
\node[nodestyle,fill=orange] (024) at (-0.95, 4.5) {$024$};
\node[nodestyle] (034) at (0.95, 4.5) {$034$};
\node[nodestyle] (123) at (2.85, 4.5) {$123$};
\node[nodestyle,fill=orange] (124) at (4.75, 4.5) {$124$};
\node[nodestyle] (134) at (6.65, 4.5) {$134$};
\node[nodestyle] (234) at (8.55, 4.5) {$234$};
\node[nodestyle] (0123) at (-3.8, 6.0) {$0123$};
\node[nodestyle] (0124) at (-1.9, 6.0) {$0124$};
\node[nodestyle,fill=orange] (0134) at (0.0, 6.0) {$0134$};
\node[nodestyle,fill=orange] (0234) at (1.9, 6.0) {$0234$};
\node[nodestyle] (1234) at (3.8, 6.0) {$1234$};
\node[nodestyle] (01234) at (0.0, 7.5) {$01234$};

 ===== DRAWING EDGES ====

\draw[black,thick] (emptyset) -- (0);
\draw[black,thick] (emptyset) -- (1);
\draw[black,thick] (emptyset) -- (2);
\draw[black,thick] (emptyset) -- (3);
\draw[black,thick] (emptyset) -- (4);
\draw[black,thick] (0) -- (01);
\draw[black,thick] (0) -- (02);
\draw[black,thick] (0) -- (03);
\draw[black,thick] (0) -- (04);
\draw[black,thick] (1) -- (01);
\draw[black,thick] (1) -- (12);
\draw[black,thick] (1) -- (13);
\draw[black,thick] (1) -- (14);
\draw[black,thick] (2) -- (02);
\draw[black,thick] (2) -- (12);
\draw[black,thick] (2) -- (23);
\draw[black,thick] (2) -- (24);
\draw[black,thick] (3) -- (03);
\draw[black,thick] (3) -- (13);
\draw[black,thick] (3) -- (23);
\draw[black,thick] (3) -- (34);
\draw[black,thick] (4) -- (04);
\draw[black,thick] (4) -- (14);
\draw[black,thick] (4) -- (24);
\draw[black,thick] (4) -- (34);
\draw[black,thick] (01) -- (012);
\draw[black,thick] (01) -- (013);
\draw[black,thick] (01) -- (014);
\draw[black,thick] (02) -- (012);
\draw[black,thick] (02) -- (023);
\draw[black,thick] (02) -- (024);
\draw[black,thick] (03) -- (013);
\draw[black,thick] (03) -- (023);
\draw[black,thick] (03) -- (034);
\draw[black,thick] (04) -- (014);
\draw[black,thick] (04) -- (024);
\draw[black,thick] (04) -- (034);
\draw[black,thick] (12) -- (012);
\draw[black,thick] (12) -- (123);
\draw[black,thick] (12) -- (124);
\draw[black,thick] (13) -- (013);
\draw[black,thick] (13) -- (123);
\draw[black,thick] (13) -- (134);
\draw[black,thick] (14) -- (014);
\draw[black,thick] (14) -- (124);
\draw[black,thick] (14) -- (134);
\draw[black,thick] (23) -- (023);
\draw[black,thick] (23) -- (123);
\draw[black,thick] (23) -- (234);
\draw[black,thick] (24) -- (024);
\draw[black,thick] (24) -- (124);
\draw[black,thick] (24) -- (234);
\draw[black,thick] (34) -- (034);
\draw[black,thick] (34) -- (134);
\draw[black,thick] (34) -- (234);
\draw[black,thick] (012) -- (0123);
\draw[black,thick] (012) -- (0124);
\draw[black,thick] (013) -- (0123);
\draw[black,thick] (013) -- (0134);
\draw[black,thick] (014) -- (0124);
\draw[black,thick] (014) -- (0134);
\draw[black,thick] (023) -- (0123);
\draw[black,thick] (023) -- (0234);
\draw[black,thick] (024) -- (0124);
\draw[black,thick] (024) -- (0234);
\draw[black,thick] (034) -- (0134);
\draw[black,thick] (034) -- (0234);
\draw[black,thick] (123) -- (0123);
\draw[black,thick] (123) -- (1234);
\draw[black,thick] (124) -- (0124);
\draw[black,thick] (124) -- (1234);
\draw[black,thick] (134) -- (0134);
\draw[black,thick] (134) -- (1234);
\draw[black,thick] (234) -- (0234);
\draw[black,thick] (234) -- (1234);
\draw[black,thick] (0123) -- (01234);
\draw[black,thick] (0124) -- (01234);
\draw[black,thick] (0134) -- (01234);
\draw[black,thick] (0234) -- (01234);
\draw[black,thick] (1234) -- (01234);
	\end{tikzpicture}
	\end{scaletikzpicturetowidth}
	\caption{The colored graph $\mathbf{G}_V[\varphi_{\mathrm{no-PM}}]$, illustrating that Conjecture~\ref{open:perfect-matchings} must be restricted to monotone functions.
	For space reasons we write, e.g., $024$ instead of $\{0,2,4\}$.}
\label{fig:non-monotone-no-PM}
\end{figure*}
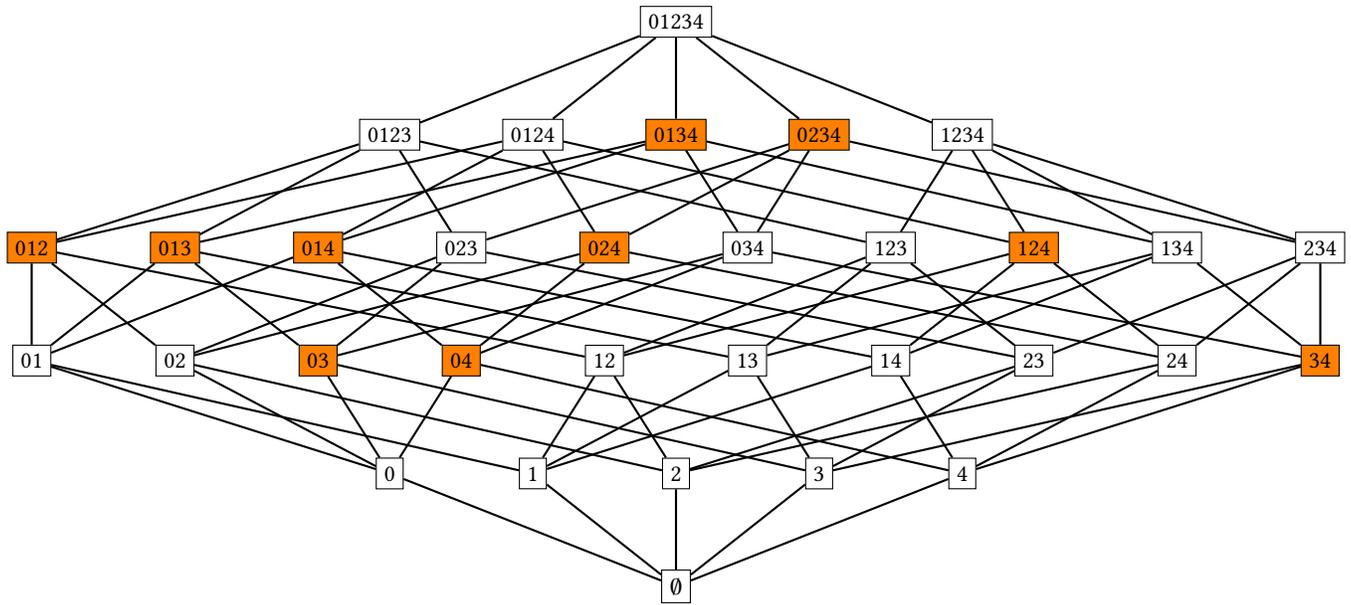
This function has zero Euler characteristic, yet it is easy to see that the subgraph induced by the colored (resp., non-colored) nodes has no perfect matching:
the colored node $\{3,4\}$ (resp., non-colored node $\{0,3,4\}$) is isolated.
In a sense, this Boolean function justifies the definition of our transformation, and shows that the approach of~\cite{monet2018towards} was doomed to fail
for the $\mathcal{H}$-queries that are not UCQs.
Second, and much harder to see, the “or” in Conjecture~\ref{open:perfect-matchings} is necessary.
Indeed, there exists a monotone function $\varphi_{\mathrm{one-neg}}$ (with $k=5$) such that the colored nodes have
no perfect matching but the non-colored nodes do (and vice versa, by symmetry).
We depicted $\mathbf{G}_V[\varphi_{\mathrm{one-neg}}]$ in Appendix~\ref{apx:open} (this is actually the smallest such function).
Third, we have checked in~\cite{monet2018towards}, using the SAT solver Glucose~\cite{audemard2009predicting}, that this conjecture holds for all
monotone Boolean functions with $k \leq 5$, amounting to about 20 million non-isomorphic (under permutation of the variables)
nondegenerate\footnote{Since for degenerate functions the conjecture clearly holds.} functions.
This conjecture seems also closely related to the open problem in knowledge compilation of determining whether d-DNNFs are closed under negation\footnote{
	Formally: given a d-DNNF $D$, is there a d-DNNF~$D'$ of polynomial size representing~$\lnot D$?}.

\begin{toappendix}
\label{apx:open}

	\begin{figure*}[h]
\centering
	\begin{scaletikzpicturetowidth}{\linewidth}
	\begin{tikzpicture}[scale=\tikzscale]
	\tikzset{nodestyle/.style={draw,rectangle}}
	\input{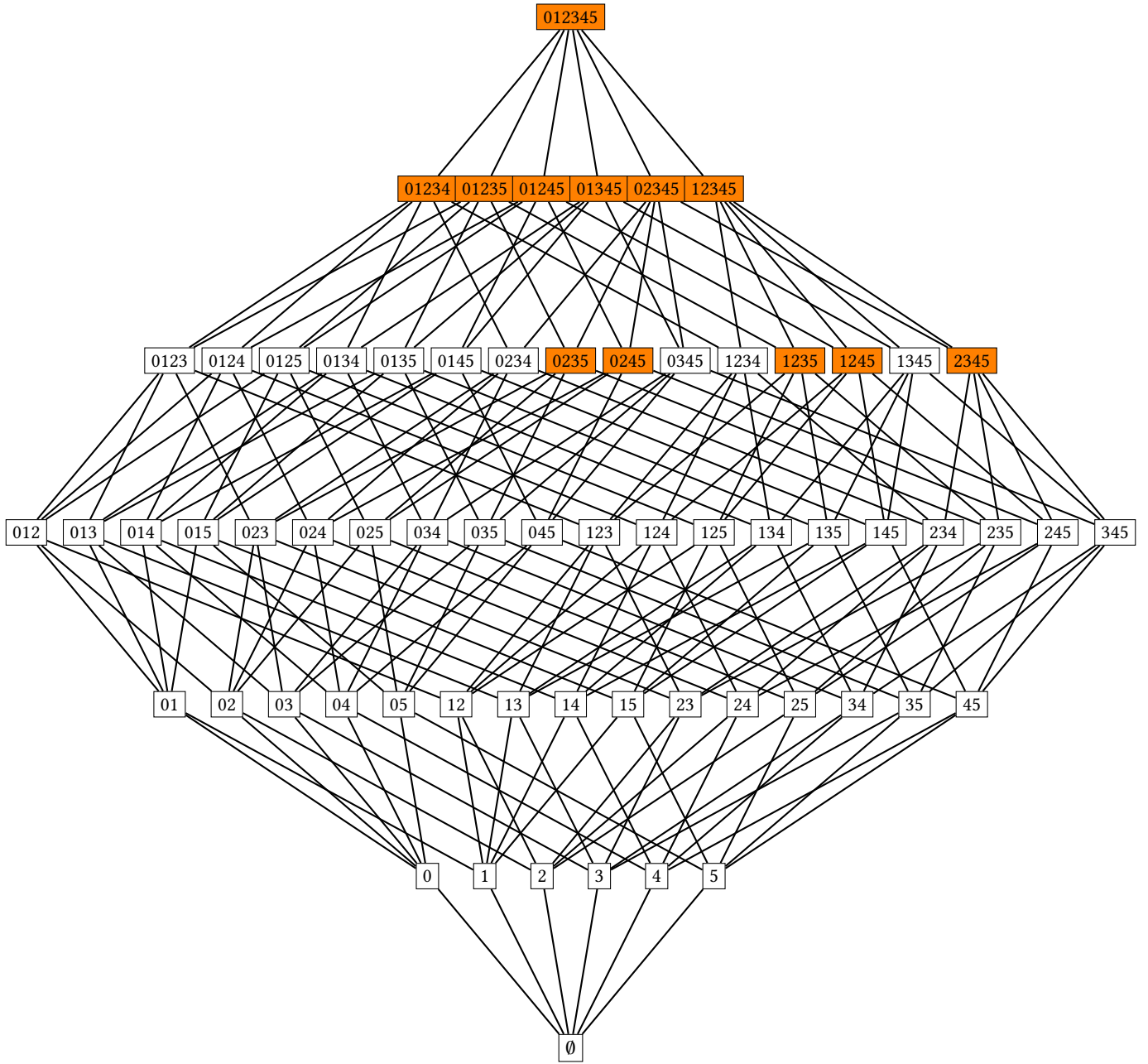}
	\end{tikzpicture}
	\end{scaletikzpicturetowidth}
	\caption{The colored graph $\mathbf{G}_V[\varphi_{\mathrm{one-neg}}]$, having $\eul(\varphi)=0$, illustrating that in Conjecture~\ref{open:perfect-matchings},
	the “or” is necessary: the colored nodes have no perfect matching (because $012345$ needs to be matched to both $01234$ and $01345$), but the non-colored ones do (checked with a SAT solver).}
\label{fig:one-neg}
\end{figure*}
\clearpage
\end{toappendix}

\paragraph{\bf Generalize our technique to all UCQs}
We now discuss the most important question: extending the techniques we have developed to capture a larger class of queries than the $\mathcal{H}$-queries.
To do this, recall that the algorithm of~\cite{dalvi2012dichotomy} for UCQs recursively alternates between two steps.
The first step is to find what is called a \emph{separator variable}, intuitively ensuring that
the subqueries obtained are independent.
This step can clearly be simulated with \mbox{d-Ds}.
The second step is to perform inclusion--exclusion using Möbius's inversion formula.
Here the algorithm recurses on the subqueries of the CNF lattice whose Möbius coefficient is not zero, potentially allowing for
\#P-hard queries to be ignored during the computation.
It is this step that seems hard to simulate using a knowledge compilation approach.
With this article, we have shown that we can simulate it, in the special case where the bottom term in the CNF lattice is the only hard subquery (and has a zero Möbius coefficient), and where we can recursively construct \mbox{d-Ds} for subqueries that
are disjunctions of two connected terms in the poset of the original query (i.e., the poset of Definition~\ref{def:induced-subgraph}).
Although we do not have a concrete example at hand\footnote{We do not think that such an example would be particularly enlightening.}, this already seems to define a larger class of queries than the $\mathcal{H}$-queries.
However, it is far from obvious how to extend our technique to avoid \emph{other} \#P-hard queries than the bottom one.
We leave this task for future work.

\begin{openpb}
	\label{open:generalize}
	Generalize our techniques to all UCQs, i.e., show that all safe UCQs are in $\ddptime$.
\end{openpb}

\section{Conclusion}
  \label{sec:conclusion}
  To the best of our knowledge, we provide the first result formally proving that the inclusion--exclusion principle can be simulated using
decomposability and determinism only.
We see this as yet another indication that knowledge compilation is an effective way to treat query answering on probabilistic databases~\cite{jha2013knowledge,amarilli2015provenance,amarilli2019connecting}.
Although this new technique applies only to a restricted class of UCQs, the queries considered here seem to already contain the core difficulty of the intensional--extensional conjecture.
We think that solving this problem for all UCQs will require solving it for UCQs with negations (more precisely, for Boolean combinations of CQs).
This is reminiscent of the algorithm of~\cite{dalvi2012dichotomy} for UCQs, which, even when applied to a CQ, can introduce UCQs in the computation.

\begin{acks}
	I acknowledge Dan Olteanu for our initial work on the problem~\cite{monet2018towards}.
	In particular, this joint work already contains the idea of covering the satisfying valuations of~$\varphi$ with mutually exclusive terms,
	which we reuse intensively here (see the paragraph “The difference with~\cite{monet2018towards}” in the introduction for
	more details).
	I also acknowledge Guy Van den Broeck, Paul Beame and Vincent Liew, who had already worked on this problem and attempted similar approaches. Guy Van den Broeck and Dan Olteanu were the first to notice that one could use negation to attack the problem. I note here that Paul Beame and his student Vincent Liew independently found the query from Figure~\ref{open:perfect-matchings} (which was also a counterexample to one of their approaches).
	I am grateful to Pierre Senellart for careful proofreading of the paper.
	I thank the Millennium Institute for Foundational Research on Data (IMFD) for funding this research.
\end{acks}

\begin{toappendix}
	\section{Overall Structure}
	\label{apx:proof-structure}
	\begin{figure}[h]
\def\svgwidth{\linewidth}
	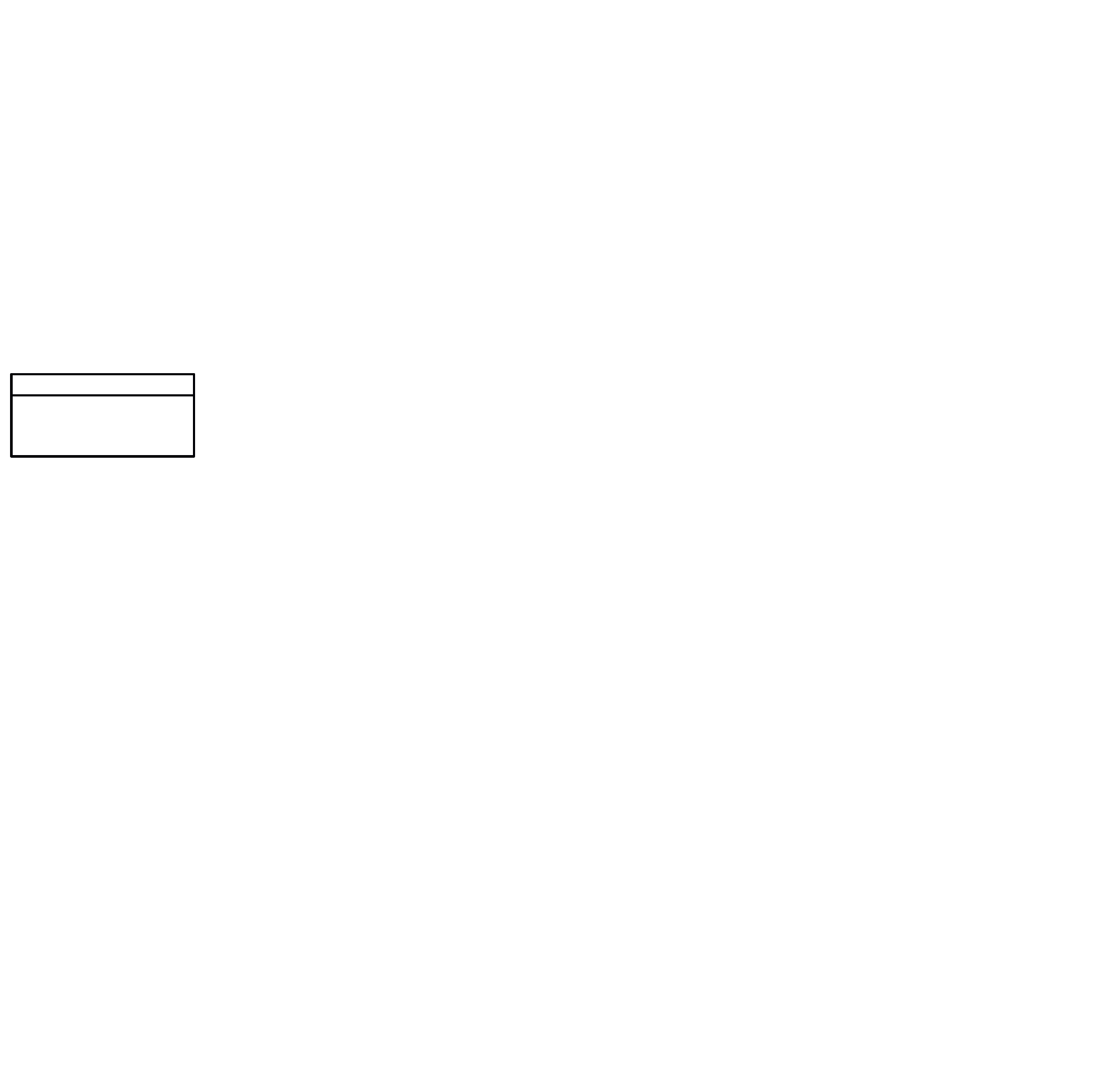 
	\caption{
		DAG representing the general structure of the proofs. Dashed regions indicate the section in which a result first appears.
		Remember that $k \in \mathbb{N}_{\geq 1}$ and the set of variables $V= \{0,\ldots,k\}$ are fixed.}
	\label{fig:proof-structure}
\end{figure}
\end{toappendix}

\bibliographystyle{ACM-Reference-Format}
\bibliography{main}

\appendix

\end{document}